\documentclass[a4paper,onecolumn,11pt,accepted=2023-03-13]{quantumarticle}
\pdfoutput=1
\usepackage[utf8]{inputenc}
\usepackage[english]{babel}
\usepackage[T1]{fontenc}
\usepackage{amsmath}
\usepackage{hyperref}
\usepackage{doi}
\usepackage{tikz}
\usepackage{lipsum}

\usepackage{commath}
\usepackage{braket}
\usepackage{multicol}
\usepackage{enumerate}
\usepackage{mathrsfs}
\usepackage{amsthm,amssymb,xspace}
\usepackage{cite}
\usepackage{txfonts}
\usepackage{graphics}
\usepackage{float}
\usepackage{epstopdf}
\usepackage{epsfig}
\usepackage{caption}
\usepackage[labelformat=simple]{subcaption}
\usepackage[pass,showframe]{geometry} 
\usepackage{pgfplots}
\usepgfplotslibrary{polar}
\usepgflibrary{shapes.geometric}
\usetikzlibrary{calc}
\usepackage[normalem]{ulem}
\usepackage{booktabs}
\usepackage{commath}
\usepackage{tikz}
\usepackage{tikz-3dplot}
\usepackage[qm]{qcircuit}
\usepackage{makecell}
\usepackage[ruled,vlined,linesnumbered]{algorithm2e}
\usepackage{multirow}

\usepackage{environ}

\newcommand{\repeattheorem}[1]{%
  \begingroup
  \renewcommand{\thetheorem}{\ref{#1}}%
  \expandafter\expandafter\expandafter\theorem
  \csname reptheorem@#1\endcsname
  \endtheorem
  \endgroup
}

\NewEnviron{reptheorem}[1]{%
  \global\expandafter\xdef\csname reptheorem@#1\endcsname{%
    \unexpanded\expandafter{\BODY}%
  }%
  \expandafter\theorem\BODY\unskip\label{#1}\endtheorem
}

\newcommand{\B}{\mbox{$\{0,1\}$}}

\newcommand\defeq{\stackrel{\mathrm{\scriptsize def}}{=}}

\newtheorem{theorem}{Theorem}
\newtheorem{corollary}[theorem]{Corollary}
\newtheorem{lemma}[theorem]{Lemma}

\newtheorem{definition}[theorem]{Definition}

\begin{document}

\title{Optimal (controlled) quantum state preparation and improved unitary synthesis by quantum circuits with any number of ancillary qubits}

\author{Pei Yuan}\thanks{peiyuan@tencent.com}
\author{Shengyu Zhang}\thanks{shengyzhang@tencent.com}
\affiliation{Tencent Quantum Laboratory, Tencent, Shenzhen, Guangdong 518057, China}

\maketitle

\begin{abstract}
As a cornerstone for many quantum linear algebraic and quantum machine learning algorithms, controlled quantum state preparation (CQSP) aims to provide the transformation of $\ket{i}\ket{0^n} \to \ket{i}\ket{\psi_i}$ for all $i\in \B^k$ for the given $n$-qubit states $\ket{\psi_i}$. 
In this paper, we construct a quantum circuit for implementing CQSP, with depth $O\left(n+k+\frac{2^{n+k}}{n+k+m}\right)$ and size $O\left(2^{n+k}\right)$ for \textit{any} given number $m$ of ancillary qubits. These bounds, which can also be viewed as a time-space tradeoff for the transformation, are \textit{optimal} for any integer parameters $m,k\ge 0$ and $n\ge 1$.

When $k=0$, the problem becomes the canonical quantum state preparation (QSP) problem with ancillary qubits, which asks for efficient implementations of the transformation $\ket{0^n}\ket{0^m} \to \ket{\psi}\ket{0^m}$. This problem has many applications with many investigations, yet its circuit complexity remains open. Our construction completely solves this problem, pinning down its depth complexity to $\Theta(n+2^{n}/(n+m))$ and its size complexity to $\Theta(2^{n})$ for \textit{any} $m$.

Another fundamental problem, unitary synthesis, asks to implement a general $n$-qubit unitary by a quantum circuit. Previous work shows a lower bound of $\Omega(n+4^n/(n+m))$ and an upper bound of $O(n2^n)$ for $m=\Omega(2^n/n)$ ancillary qubits. In this paper, we quadratically shrink this gap by presenting a quantum circuit of the depth of  $O\left(n2^{n/2}+\frac{n^{1/2}2^{3n/2}}{m^{1/2}}\right)$. 
\end{abstract}

\section{Introduction}\label{sec:intro}

Quantum algorithms use quantum effects such as quantum entanglement and coherence to process information with the efficiency beyond any classical counterparts can achieve. In the past decade, 
many quantum machine learning algorithms \cite{biamonte2017quantum}  share a common subroutine of \textit{quantum state preparation} (QSP), which loads a $2^n$-dimensional complex-valued vector $v=(v_x: x\in \B^n)^T\in \mathbb{C}^{2^n}$ to an $n$-qubit quantum state $|\psi_v\rangle =\sum_{x\in\{0,1\}^n}v_x|x\rangle$. These include quantum principle component analysis \cite{lloyd2014quantum}, quantum recommendation systems \cite{kerenidis2017quantum}, quantum singular value decomposition \cite{rebentrost2018quantum}, quantum linear system algorithm \cite{harrow2009quantum,wossnig2018quantum}, quantum clustering \cite{kerenidis2019q,kerenidis2020quantum}, quantum support vector machine \cite{rebentrost2014quantum}, etc. Quantum state preparation is also a key step in many Hamiltonian simulation algorithms \cite{berry2015simulating,low2017optimal,low2019hamiltonian,berry2015hamiltonian}.

{Some of these quantum machine learning algorithms, such as quantum linear system algorithm \cite{wossnig2018quantum}, quantum recommendation systems \cite{kerenidis2017quantum}
and quantum $k$-means clustering \cite{kerenidis2019q}, need an oracle that can \textit{coherently} prepare many states: $\ket{i}\ket{0^{n}}\to\ket{i}\ket{\psi_i},\text{~for~all~}i\in \{0,1\}^k$. 
We shall refer to this as the \textit{controlled quantum state preparation} (CQSP) problem.} 
The QSP and CQSP problems are also used in quantum walk algorithms such as the one by Szegedy \cite{szegedy2004quantum} and by MNRS \cite{magniez2011search}.  Given a general $N\times N$ state transition probability matrix $P = [P_{xy}]_{x,y\in [N]}$, quantum walk algorithms often call three subroutines: Setup, Check, and Update. 
The Setup procedure needs to prepare state $\sum_x \sqrt{\pi(x)}\ket{x}$, where $\pi$ is the stationary distribution of $P$. The Update procedure needs to realize $\ket{x} \ket{0^{\log N}} \to \ket{x} \sum_y \sqrt{P_{xy}}\ket{y}$, a typical CQSP problem.

More generally, quantum algorithms can be represented as unitaries, which need to be implemented by quantum circuits for a digital quantum computer to run the algorithm. What is the minimum depth and size that any unitary operation can be compressed to? This paper also addresses this Unitary synthesis (US) problem by presenting a parametrized quantum circuit that can implement any given unitary operation. 

In all these CQSP, QSP, and US problems, we hope to find quantum circuits as simple as possible for the sake of efficiency of execution and physical realization. Standard measures for quantum circuits include depth, size (i.e. the number of gates), and a number of qubits. The depth of a circuit corresponds to time and the number of qubits to space. The rapid advancement of the number of qubits provides opportunities to trade space for time, and indeed it has been found that ancillary qubits are useful in compressing the circuit depth for many tasks including CQSP, QSP, and US. It is a fundamental question to pin down the time-space tradeoff, or in circuit complexity language, the depth-qubit number tradeoff, for both quantum state preparation and general unitary synthesis problems.


\paragraph{Controlled quantum state preparation and Quantum state preparation}
Much previous work focuses on specific CQSP by quantum circuits \cite{park2019circuit,de2020circuit,di2020fault}.
QSP, in contrast, has been extensively studied. Bergholm \textit{et al.} presented a QSP circuit with $2^{n+1}-2n-2$ CNOT gates and depth $O(2^n)$, without ancillary qubits \cite{bergholm2005quantum}. Plesch and Brukner  \cite{plesch2011quantum} improve the number of CNOT gate to $\frac{23}{24} 2^n - 2^{\frac{n}{2}+1} + \frac{5}{3}$ for even $n$, and $\frac{115}{96} 2^n $ for odd $n$. Ref. \cite{bergholm2005quantum} also gives a depth upper bound of $\frac{23}{48}2^n$ for even $n$ and $\frac{115}{192}2^n$ for odd $n$. The best result was obtained in \cite{sun2021asymptotically}, where the authors achieve the depth $O(2^n/n)$, which is optimal. 

Zhang \textit{et al.} \cite{zhang2021low} presented a QSP circuit of depth $O(n^2)$, by using $O(4^n)$ ancillary qubits, but the circuit involves measurement and the  probability of successfully generating the target state is only $\Omega(1 / (\max_i |v_i|^2 2^n))$. 

The best previous result on QSP for an arbitary number $m$ of ancillary qubits is by \cite{sun2021asymptotically}, where the authors presented a quantum circuit of depth $O\left(n+\frac{2^n}{n+m}\right)$ and size $O(2^n)$ for $m=O\left(\frac{2^n}{n\log n}\right)$ or $m=\Omega(2^n)$, which is asymptotically optimal. For $m\in\left[\Omega\left(\frac{2^n}{n\log n}\right),O(2^n)\right]$, they proposed a QSP circuit of depth $O(n\log n)$, which is only $O(\log n)$ off from the lower bound $\Omega\left(\max\left\{n,\frac{2^n}{n+m}\right\}\right)$. Later, Rosenthal independently constructed a QSP circuit of depth $O(n)$ using ${O}(n2^n)$ ancillary qubits \cite{rosenthal2021query}. The result also showed that an $n$-qubit quantum state preparation is in ${\sf QAC}_f^0$ with the same number of ancillary qubits. After that, \cite{zhang2022quantum} gave yet another proof of the $O(n)$ depth upper bound using $O(2^n)$ ancillary qubits. Both \cite{rosenthal2021query} and \cite{zhang2022quantum} did not give results for general $m$.

A related study is to prepare a quantum state in the \textit{unary} encoding $\sum_{k=0}^{2^k-1}v_k\ket{e_k}$ instead of binary encoding $\sum_{k=1}^{2^n-1}v_k\ket{k}$ in \cite{johri2021nearest}, where $e_i\in \{0,1\}^{2^n}$ is the vector with the $k$-th bit being 1 and all other bits being 0. The binary encoding quantum state preparation is more efficient than unary encoding because binary encoding QSP utilizes $n$ qubits but unary encoding utilizes $2^n$ qubits. In \cite{johri2021nearest}, Johri \textit{et al.} prepared a unary encoding quantum state by a circuit of depth $O(n)$. Moreover, by encoding $k$ to a $d$-dimensional tensor $(k_1, k_2, \ldots, k_d)$, they extended the QSP circuit construction and obtained circuit depth $O\left(\frac{n}{d}2^{n-n/d}\right)$. If $d=n$, their encoding of $k$ is binary encoding, and the depth upper bound is $O(2^n)$.

In this paper, we first give new quantum circuit constructions for CQSP with quantum content.

\begin{theorem}
[CQSP] \label{thm:multi-QSP}
For any integers $k,m\ge 0$, $n>0$ and any quantum states $\{\ket{\psi_i}:i\in \{0, 1\}^k\}$, the following controlled quantum state preparation
\[\ket{i}\ket{0^{n}}\to\ket{i}\ket{\psi_i}, \ \forall i\in \{0, 1\}^k\]
can be implemented by a quantum circuit consisting of single-qubit and CNOT gates of depth $O\left(n+k+\frac{2^{n+k}}{n+k+m}\right)$ and size $O\left(2^{n+k}\right)$ with $m$ ancillary qubits. These bounds are optimal for any $k,m \ge 0$.
\end{theorem}

Taking $k=0$, This immediately implies the following result for QSP.
\begin{theorem}[QSP]
\label{thm:QSP}
For any $m\ge 0$, any $n$-qubit quantum state $\ket{\psi_v}$ can be generated by a quantum circuit, using single-qubit gates and CNOT gates, of depth $O\left(n+\frac{2^n}{n+m}\right)$ and size $O(2^{n})$  with $m$ ancillary qubits. These bounds are optimal for any $m \ge 0$.
\end{theorem}

These bounds match the known lower bounds of circuit depth and size for QSP: $\Omega\big(\max \big\{n,\frac{4^n}{n+m}\big\}\big)$ for depth \cite{zhang2021parallel,sun2021asymptotically} and $\Omega(4^n)$ for size \cite{shende2004minimal}. 
Thus we completely characterize the depth and size complexity for QSP with any number $m$ of ancillary qubits. 

\paragraph{Unitary synthesis}
For general unitary synthesis, Barenco \textit{et al.} constructed a circuit involving $O(n^3 4^n)$ CNOT gates \cite{barenco1995elementary}. Knill reduced the circuit size to $O(n4^n)$ in \cite{knill1995approximation}, which was further improved by Vartiainen \textit{et al.} \cite{vartiainen2004efficient} and Mottonen and Vartiainen \cite{mottonen2005decompositions} to $O(4^n)$, the same order as the lower bound of $\left\lceil\frac{1}{4}(4^n-3n-1)\right\rceil$ for for the number of CNOT gates \cite{shende2004minimal}.
 
These results assume no ancillary qubits. When there are $m$ ancillary qubits available, Ref. \cite{sun2021asymptotically} presented a quantum circuit for $n$-qubit general unitary synthesis of depth $O\left(n2^n+\frac{4^n}{n+m}\right)$, and also proved a depth lower bound of $\Omega\left(n+\frac{4^n}{n+m}\right)$. Hence, their circuit depth bounds are asymptotically optimal when $m=O\left(2^n/n\right)$, and leave a gap of $\left[\Omega\left(n+\frac{4^n}{m}\right),O\left( n2^n+\frac{4^n}{m}\right)\right]$ when $m=\Omega\left(2^n/n\right)$. By using Grover search in a clever way, Rosenthal improved the depth upper bound to $O(n2^{n/2})$ with $m = \Theta(n4^n)$ ancillary qubits \cite{rosenthal2021query}, but did not give results for smaller $m$.

For general unitary synthesis, based on the cosine-sine decomposition and Grover search, we can improve the circuit depth for general unitary synthesis as follows.
\begin{theorem}[Unitary synthesis]\label{thm:unitary_mul_intro}
For any $m\ge 0$, any $n$-qubit unitary $U\in\mathbb{C}^{2^n\times 2^n}$ can be implemented by a quantum circuit with $m$ ancillary qubits, using single-qubit gates and CNOT gates, { of depth $O(4^n/(n+m))$} when $m=O(2^n/n)$, $O\left(\frac{n^{1/2}2^{3n/2}}{m^{1/2}}\right)$ when ${\Omega(2^n/n) \le m}\le O(4^n/n)$, and $O(n2^{n/2})$ when $m = \Omega(4^n/n)$.
\end{theorem}
The improvement of this result over the one in \cite{rosenthal2021query} is two-fold. First, to achieve the same minimum depth of $O(n2^{n/2})$, we need fewer ancillary qubits: we need $m = \Theta(4^n/n)$ compared to $m = \Theta(n4^n)$ used in \cite{rosenthal2021query}. Second, our method works for any $m$ as opposed to the one in \cite{rosenthal2021query} which needs $m = \Theta(n4^n)$ many ancillary qubits. Note that there is still a gap between upper and lower bounds when $m = \omega(2^n/n)$, left as an interesting open question for future studies.

Theorem \ref{thm:unitary_mul_intro} and previous results on circuit depth for general unitary synthesis are shown in Figure \ref{fig:result}.

\begin{figure*}[!ht]
\centering
\begin{tikzpicture}

\draw[->] (0,0) -- (8,0) node[anchor=north] {\small $\sharp$ ancilla~$m$};
		
\draw (0,0) node[anchor=east] {$O$};
\draw (0,6) node[anchor=east] {\scriptsize $O(4^n/n)$};
\draw (0,36/21.4) node[anchor=east] {\scriptsize $O(n2^n)$};
\draw (0,1) node[anchor=east] {\scriptsize $O(n2^{n/2})$};
\draw (0,0.5) node[anchor=east] {\scriptsize $O(n)$};
\draw (1.4, 0) node[anchor=north] {\scriptsize $O(\frac{2^n}{n})$};
\draw (6, 0) node[anchor=north] {\scriptsize $O(\frac{4^n}{n})~~~$};
\draw (6.5, 0) node[anchor=north] {\scriptsize $~~~~~O(n4^n)$};

\draw[->] (0,0) -- (0,7) node[anchor=east] {\small circuit depth $d$};

\draw	(3,4.5) node[anchor=west]{\small Depth upper bound in \cite{sun2021asymptotically}}
		(3,5) node[anchor=west]{\small Depth lower bound in \cite{sun2021asymptotically} }
		(3,5.5) node[anchor=west]{\small Depth upper bound in Theorem \ref{thm:unitary_mul_intro}}
		(3,4) node[anchor=west]{\small Depth upper bound in \cite{rosenthal2021query}};
\draw[thick,blue] (2.5,4.5)--(3,4.5);
\draw[thick, orange] (2.5,5)--(3,5);
\draw[thick,purple] (2.5,5.5)--(3,5.5);
\draw[thick] (2.3,3.6) -- (8.5,3.6) -- (8.5,5.9) -- (2.3,5.9) -- cycle;

\draw [green, thick] (6.5,1.03) -- (8,1.04);
\draw[line width=1.5pt, thick, orange] plot[smooth, domain = 0:1.4] (\x, {(36/11)/((\x+0.05) + (6/11))+0.02});
\draw[thick,orange] plot[smooth, domain = 1.4:6] (\x, {(36/11)/(\x + (6/11))});
\draw[thick,orange] (6,0.5) -- (8,0.5);

\draw[blue]	(6.2,1.95) node[anchor=west]{\scriptsize $O\left(n2^n+\frac{4^n}{n+m}\right)$};
\draw[purple]	(5.6,1.3) node[anchor=west]{\scriptsize $O\left(n2^{n/2}+\frac{n^{1/2}2^{3n/2}}{m^{1/2}}\right)$};
\draw[orange]	(6,0.75) node[anchor=west]{\scriptsize $\Omega\left(\max\{n,\frac{4^n}{n+m}\}\right)$};

\draw[thick, blue] plot[smooth, domain = 0:1.4] (\x, {(36/11)/(\x + (6/11))});
\draw[thick, blue] (1.4,36/21.4) -- (8,36/21.4);

\draw[thick,purple] plot[smooth, domain = 1.4:6] (\x, {(165.6/14.6)/(\x + (78/14.6))});
\draw[thick,purple] (6,1) -- (8,1);

\draw[dotted, gray] (0,1) -- (6, 1) -- (6,0);
\draw[dotted, gray] (0,36/21.4) -- (1.4, 36/21.4) -- (1.4,0) ;
\draw[dotted, gray] (0,0.5) -- (6, 0.5);
\draw[dotted, gray] (6.5, 1) -- (6.5, 0);

\draw (6.5,1) node[fill,black,draw=black,circle,scale = 0.2]{};
\draw (0,6) node[fill,black,draw=black,circle,scale = 0.2]{};
\draw (1.4, 36/21.4) node[fill,black,draw=black,circle,scale = 0.2]{};
\draw (6,0.5) node[fill,black,draw=black,circle,scale = 0.2]{};
\draw (6,1) node[fill,black,draw=black,circle,scale = 0.2]{};
\draw (0,36/21.4) node[fill,black,draw=black,circle,scale = 0.2]{};
\draw (0,1) node[fill,black,draw=black,circle,scale = 0.2]{};
\draw (0,1/2) node[fill,black,draw=black,circle,scale = 0.2]{};
\draw (1.4,0) node[fill,black,draw=black,circle,scale = 0.2]{};
\draw (6,0) node[fill,black,draw=black,circle,scale = 0.2]{};
\draw (6.5,0) node[fill,black,draw=black,circle,scale = 0.2]{};

\draw [green,thick] (2.5,4) -- (3,4);
\end{tikzpicture}
\caption{ Circuit depth upper and lower bounds for general $n$-qubit unitary when $m$ ancillary qubits are available. Ref. \cite{sun2021asymptotically} gives an  upper bound of $O\left(n2^n+\frac{4^n}{n+m}\right)$ and a lower bound of $\Omega\left(\max\left\{n+\frac{4^n}{n+m}\right\}\right)$, Ref. \cite{rosenthal2021query} presents a quantum circuit of depth $O\left(n2^{n/2}\right)$ using  $m=\Theta(n4^n)$ ancillary qubits. This paper gives an upper bound of $O\left(n2^{n/2}+\frac{n^{1/2}2^{3n/2}}{m^{1/2}}\right)$ for any $m=\Omega(2^n/n)$.
} 
\label{fig:result}
\end{figure*}
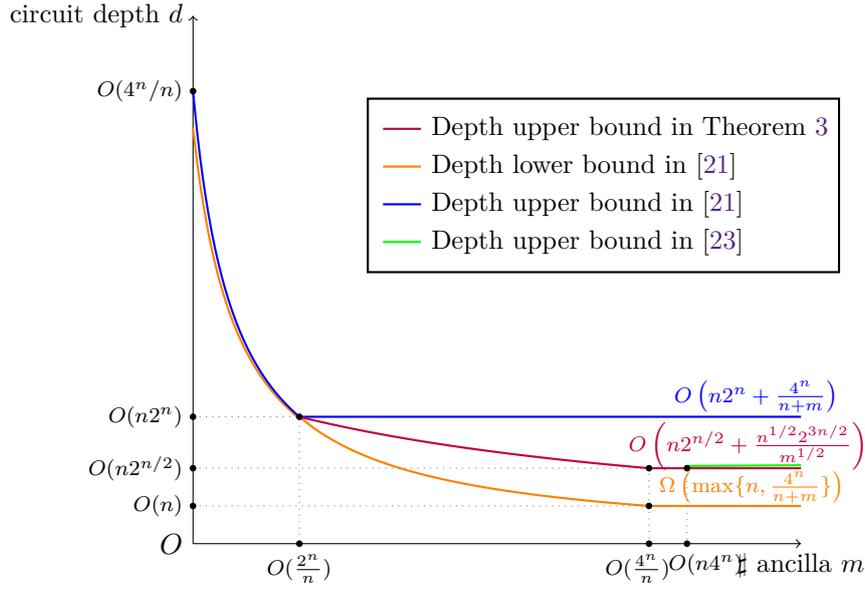

\medskip
\noindent { \textbf{Relation to QRAM}. The CQSP problem has a close relation to quantum random access memory (QRAM). In the original proposal \cite{giovannetti2008quantum,Giovannetti2008Architectures}, QRAM aims to provide the transformation of $\ket{i}\ket{0^n} \to \ket{i}\ket{\psi_i}$ for all $i\in \B^k$, where $\ket{\psi_i}$'s are states in $\{\ket{0},\ket{1}\}^{\otimes n}$ or in $(\mathbb{C}^2)^{\otimes n}$, depending on whether the QRAM stores classical or quantum information as its content. Many quantum algorithms such as those mentioned at the beginning of this section, usually assume an efficient implementation of QRAM with classical content, and the hope is to have a hardware device to realize this. Despite some conceptual designs, working QRAM devices are yet to be convincingly demonstrated, even for a small scale. Results on CQSP in this paper address a related and fundamental question of \textit{implementing QRAM (with quantum content) by standard quantum circuits}, and show tight depth and size bounds for it.} 

\paragraph{Organization.} The rest of this paper is organized as follows. In Section \ref{sec:pre}, we introduce notation and review some previous results.
In Section \ref{sec:QSP_optimal}, we will present a quantum circuit for (controlled) quantum state generation using arbitrary number of ancillary qubits. Then we will show a quantum circuit for general unitary synthesis in Section \ref{sec:unitary}.

\section{Preliminary}\label{sec:pre}
\paragraph{Notation} Let $[n]$ denote the set $\{1,2,\cdots,n\}$. All logarithms $\log(\cdot)$ are base 2 in this paper.  
For any $x=x_1\cdots x_s\in\{0,1\}^s,y=y_1\cdots y_t\in\{0,1\}^t$, $xy$ denotes the $(s+t)$-bit string $x_1\cdots x_s y_1\cdots y_t\in\{0,1\}^{s+t}$. 
The $n$-qubit state $\ket{i} = \ket{i_0 i_1 \cdots i_{n-1}} \in (\{\ket{0}, \ket{1}\})^{\otimes n}$ is the \textit{binary} encoding of $i$ satisfying $i=\sum_{j=0}^{n-1} i_j\cdot 2^j$.

\paragraph{Quantum gates and circuits} 
An $n$-qubit gate/unitary is a $2^n\times 2^n$ unitary operation on $n$ qubits. The identity unitary is usually denoted by $\mathbb{I}_n$. The $X$ gate is the single-qubit gate that flips the basis $\ket{0}$ and $\ket{1}$. Single-qubit gates are known to have the following factorization. 
\begin{lemma}[\cite{nielsen2002quantum}, Corollary 4.2]\label{lem:1Qfactor}
    Any single-qubit gate $U$ can be written as $U = e^{i\alpha} AXBXC$ for some $\alpha\in \mathbb{R}$ and some single-qubit gate $A$, $B$ and $C$ satisfying $ABC = \mathbb{I}_1$. 
\end{lemma}

A CNOT gate acts on two qubits, one control qubit, and one target qubit. The gate flips the basis $\ket{0}$ and $\ket{1}$ on the target qubit, conditioned on the control qubit is on $\ket{1}$. A quantum circuit on $n$ qubits implements a unitary transform of dimension $2^n\times 2^n$. A quantum circuit may consist of different types of gates. One typical set of gates contains all 1-qubit gates and 2-qubit CNOT gates. This is sufficient to implement any unitary transform. For notational convenience, we call this type of quantum circuits the \textit{standard quantum circuits}. Unless otherwise stated, a circuit in this paper refers to a standard quantum circuit. A subset of circuits is \textit{CNOT circuits}, which are the ones consisting of 2-qubit CNOT gates only. 

A Toffoli gate is a 3-qubit CCNOT gate where we flip the basis $\ket{0},\ket{1}$ of (i.e. apply $X$ gate to) the third qubit conditioned on the first two qubits are both on $\ket{1}$. Namely, there are two control qubits and one target qubit. This can be extended to an $n$-fold Toffoli gate, which applies the $X$ gate to the $(n+1)$-th qubit conditioned on the first $n$ qubits all being on $\ket{1}$. This $n$-fold Toffoli gate can be implemented by a standard quantum circuit of linear depth and size without ancillary qubits \cite{multi-controlled-gate}, and of logarithmic depth and linear size if a linear number of ancillary qubits are available \cite{baker2019decomposing}.
\begin{lemma} \label{lem:tof}
An $n$-fold Toffoli gate can be implemented by a standard quantum circuit of $O(n)$ depth and size without using any ancillary qubit and to $O(\log n)$ depth and $O(n)$ size using $O(n)$ ancillary qubits.
\end{lemma}

A non-standard quantum circuit model is $\textsf{QAC}^0_f$ circuit. A $\textsf{QAC}^0_f$ circuit is a quantum circuit with one-qubit gates, unbounded-arity Toffoli 
\[\ket{x_1, \ldots, x_k, b} \to \ket{x_1, \ldots, x_k, b\oplus \prod_{i=1}^k x_i}, \]
and 
fanout gates
\[\ket{b, x_1, \ldots, x_k} \to \ket{b, x_1\oplus b, \ldots, x_k\oplus b}.\]

\paragraph{CQSP,\text{ } QSP,\text{ } and US problems}
\begin{enumerate}
    \item The Controlled Quantum State Preparation (CQSP) problem is: Given $2^k$ quantum states $\ket{\psi_i}$ of $n$ qubits, realize the transformation of 
    \[\ket{i}\ket{0^n} \to \ket{i}\ket{\psi_i}, \ \forall i\in \B^k.\] 
    We sometimes write $(k,n)$-CQSP to emphasize the parameters. 
    
    \item The Quantum State Preparation (QSP) problem is the above CQSP problem in the special case of $k=0$. Given a complex vector $v=(v_0,v_1,v_2,\ldots,v_{2^n-1})^T\in \mathbb{C}^{2^n}$ with  $\sqrt{\sum_{j=0}^{2^n-1}|v_j|^2}=1$, generate the corresponding $n$-qubit quantum state \[\ket{\psi_v}=\sum_{j=0}^{2^n-1}v_j|j\rangle,\]
    by a quantum circuit from the initial state $\ket{0}^{\otimes n}$, where $\{|j\rangle: j=0, 1, \ldots, 2^n-1\}$ is the computational basis of the quantum system. We sometimes call a quantum circuit for quantum state preparation a QSP circuit. 
    
    \item The general Unitary Synthesis (US) problem is: Given an $n$-qubit unitary $U$, find a quantum circuit to implement it. 
\end{enumerate}

In all these problems, we hope to find circuits as simple as possible, and standard measures for quantum circuits include depth, size (i.e. the number of gates), and number of qubits. The depth of a circuit corresponds to time and the number of qubits to space. For many information processing tasks including QSP and US, ancillary qubits turn out to be very helpful, and indeed there have been studies on quantum circuits with ancillary qubits for QSP and US. Since these tasks are often used as subroutines, it is usually desirable to have the ancillary qubits initialized to $\ket{0}$ at the beginning and are restored  to $\ket{0}$ at the end. Thus we say that a quantum circuit $C$ prepares an $n$-qubit quantum state $\ket{\psi}$ with $m$ ancillary qubits if 
\[{C}\left(\ket{0}^{\otimes n}\ket{0}^{\otimes m}\right) =\ket{\psi}\ket{0}^{\otimes m}.\] 
Similarly, we call an $(n+m)$-qubit quantum circuit ${C}$ implements an $n$-qubit unitary $U$ 
using $m$ ancillary qubits if \[{C}\left(\ket{\psi}\ket{0}^{\otimes m}\right) = \left(U\ket{\psi}\right)\otimes \ket{0}^{\otimes m},\ \text{for any  $n$-qubit state } \ket{\psi}.\]

\paragraph{Uniformly Controlled Unitary (UCU)} Let $S=\{s_1,\ldots,s_k\}$, $T=\{t_1,\ldots,t_\ell\}$ and $S\cap T=\emptyset$. A \textit{uniformly controlled unitary} $V^S_T$ consists of $2^k$ controlled unitary operations, where $S$ is the index set of the control qubits, and $T$ is the index set of target qubits. The $2^k$ multiple-controlled unitary operations are conditioned on distinct basis states of the $k$ control qubits; see Figure \ref{fig:UCU} for the circuit representation of $V^S_T$, where $U$ is a shorthand for the collection of $U_0, U_1, \ldots, U_{2^k-1}$. To make the sizes of $S$ and $T$ explicit, we sometimes call  $V^S_T$ a $(k,\ell)$-UCU. The matrix representation of $V_T^S$ is
\[
V^S_T=\left(\begin{array}{cccc}
    U_0 &  & &\\
     & U_1 & &\\
      &  &\ddots &\\
       &  & & U_{2^k-1}\\
\end{array}\right)\in\mathbb{C}^{2^{(k+\ell)}\times 2^{(k+\ell)}},
\]
where $U_0,U_1,\ldots,U_{2^k-1}\in \mathbb{C}^{2^{\ell}\times 2^{\ell}}$ are unitary matrices.  If $S=\emptyset$, $V^S_T$ is a just an $\ell$-qubit unitary operation. If $\ell = 1$, the UCU is also called \textit{uniformly controlled gate} (UCG), and we refer to $k$-UCG for $(k,1)$-UCU.
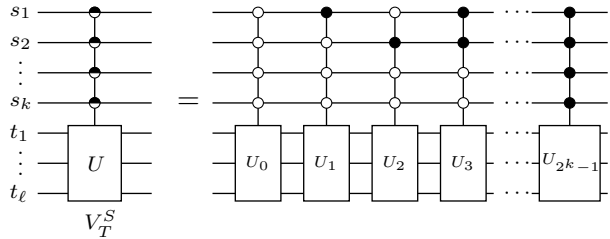
\begin{figure}[H]
\centering
\begin{tikzpicture}
\draw[line width =0.5pt] (0,0) -- (1.5,0);
\draw[line width =0.5pt] (0,-0.4) -- (1.5,-0.4);
\draw[line width =0.5pt] (0,-0.8) -- (1.5,-0.8);
\draw[line width =0.5pt] (0,-1.2) -- (1.5,-1.2);
\draw[line width =0.5pt] (0,-1.6) -- (1.5,-1.6);
\draw[line width =0.5pt] (0,-2) -- (1.5,-2);
\draw[line width =0.5pt] (0,-2.4) -- (1.5,-2.4);

\draw[line width =0.5pt] (0.75,0) -- (0.75,-2.4);

\draw [fill=white] (0.75,0) circle (0.07);
\draw [fill=black] (0.68,0) arc (180:0:0.07) -- cycle;
\draw [fill=white] (0.75,-0.4) circle (0.07);
\draw [fill=black] (0.68,-0.4) arc (180:0:0.07) -- cycle;
\draw [fill=white] (0.75,-0.8) circle (0.07);
\draw [fill=black] (0.68,-0.8) arc (180:0:0.07) -- cycle;
\draw [fill=white] (0.75,-1.2) circle (0.07);
\draw [fill=black] (0.68,-1.2) arc (180:0:0.07) -- cycle;

\draw [fill=white, line width =0.5pt] (0.4, -1.5) -- (1.1,-1.5) -- (1.1,-2.5) -- (0.4,-2.5) -- cycle;

\draw (0.52,-2) node[anchor=west]{\scriptsize $U$};

\draw (-0.5,0) node[anchor=west]{\scriptsize $s_1$};
\draw (-0.5,-0.4) node[anchor=west]{\scriptsize $s_2$};
\draw (-0.4,-0.7) node[anchor=west]{\scriptsize $\vdots$};
\draw (-0.5,-1.2) node[anchor=west]{\scriptsize $s_k$};
\draw (-0.5,-1.6) node[anchor=west]{\scriptsize $t_1$};
\draw (-0.4,-1.9) node[anchor=west]{\scriptsize $\vdots$};
\draw (-0.5,-2.4) node[anchor=west]{\scriptsize $t_\ell$};

\draw[line width =0.5pt] (2.3,0) -- (6,0);
\draw[line width =0.5pt] (2.3,-0.4) -- (6,-0.4);
\draw[line width =0.5pt] (2.3,-0.8) -- (6,-0.8);
\draw[line width =0.5pt] (2.3,-1.2) -- (6,-1.2);
\draw[line width =0.5pt] (2.3,-1.6) -- (6,-1.6);
\draw[line width =0.5pt] (2.3,-2) -- (6,-2);
\draw[line width =0.5pt] (2.3,-2.4) -- (6,-2.4);

\draw[line width =0.5pt] (6.5,0) -- (7.5,0);
\draw[line width =0.5pt] (6.5,-0.4) -- (7.5,-0.4);
\draw[line width =0.5pt] (6.5,-0.8) -- (7.5,-0.8);
\draw[line width =0.5pt] (6.5,-1.2) -- (7.5,-1.2);
\draw[line width =0.5pt] (6.5,-1.6) -- (7.5,-1.6);
\draw[line width =0.5pt] (6.5,-2) -- (7.5,-2);
\draw[line width =0.5pt] (6.5,-2.4) -- (7.5,-2.4);

\draw (1.7,-1.2) node[anchor=west]{$=$};

\draw[line width =0.5pt] (2.9,0) -- (2.9,-2.4);
\draw [fill=white] (2.9,0) circle (0.07);
\draw [fill=white] (2.9,-0.4) circle (0.07);
\draw [fill=white] (2.9,-0.8) circle (0.07);
\draw [fill=white] (2.9,-1.2) circle (0.07);
\draw [fill=white, line width =0.5pt] (2.6, -1.5) -- (3.2,-1.5) -- (3.2,-2.5) -- (2.6,-2.5) -- cycle;
\draw (2.6,-2) node[anchor=west]{\tiny$U_0$};

\draw[line width =0.5pt] (3.8,0) -- (3.8,-2.4);
\draw [fill=black] (3.8,0) circle (0.07);
\draw [fill=white] (3.8,-0.4) circle (0.07);
\draw [fill=white] (3.8,-0.8) circle (0.07);
\draw [fill=white] (3.8,-1.2) circle (0.07);
\draw [fill=white, line width =0.5pt] (3.5, -1.5) -- (4.1,-1.5) -- (4.1,-2.5) -- (3.5,-2.5) -- cycle;
\draw (3.5,-2) node[anchor=west]{\tiny $U_1$};

\draw[line width =0.5pt] (4.7,0) -- (4.7,-2.4);
\draw [fill=white] (4.7,0) circle (0.07);
\draw [fill=black] (4.7,-0.4) circle (0.07);
\draw [fill=white] (4.7,-0.8) circle (0.07);
\draw [fill=white] (4.7,-1.2) circle (0.07);
\draw [fill=white, line width =0.5pt] (4.4, -1.5) -- (5,-1.5) -- (5,-2.5) -- (4.4,-2.5) -- cycle;
\draw (4.4,-2) node[anchor=west]{\tiny $U_2$};

\draw[line width =0.5pt] (5.6,0) -- (5.6,-2.4);
\draw [fill=black] (5.6,0) circle (0.07);
\draw [fill=black] (5.6,-0.4) circle (0.07);
\draw [fill=white] (5.6,-0.8) circle (0.07);
\draw [fill=white] (5.6,-1.2) circle (0.07);
\draw [fill=white, line width =0.5pt] (5.3, -1.5) -- (5.9,-1.5) -- (5.9,-2.5) -- (5.3,-2.5) -- cycle;
\draw (5.3,-2) node[anchor=west]{\tiny $U_3$};

\draw (6,0) node[anchor=west]{\scriptsize $\cdots$};
\draw (6,-0.4) node[anchor=west]{\scriptsize $\cdots$};
\draw (6,-0.8) node[anchor=west]{\scriptsize $\cdots$};
\draw (6,-1.2) node[anchor=west]{\scriptsize $\cdots$};
\draw (6,-1.6) node[anchor=west]{\scriptsize $\cdots$};
\draw (6,-2) node[anchor=west]{\scriptsize $\cdots$};
\draw (6,-2.4) node[anchor=west]{\scriptsize $\cdots$};

\draw[line width =0.5pt] (7,0) -- (7,-2.4);
\draw [fill=black] (7,0) circle (0.07);
\draw [fill=black] (7,-0.4) circle (0.07);
\draw [fill=black] (7,-0.8) circle (0.07);
\draw [fill=black] (7,-1.2) circle (0.07);
\draw [fill=white, line width =0.5pt] (6.6, -1.5) -- (7.4,-1.5) -- (7.4,-2.5) -- (6.6,-2.5) -- cycle;
\draw (6.5,-2) node[anchor=west]{\tiny $U_{2^k-1}$};
\draw (0.48,-2.8) node[anchor=west]{\scriptsize $V^S_T$};
\end{tikzpicture}
\caption{A uniformly controlled unitary (UCU) $V^{S}_T$, where $S=\{s_1,\ldots,s_k\}$ is the index set of the control qubits and $T=\{t_1,\ldots,t_\ell\}$ is the index set of the target qubits.}
\label{fig:UCU}
\end{figure}

Ref. \cite{sun2021asymptotically} gives the following size and depth upper bounds of a general UCG, which is a special case of our later Lemma \ref{lem:multi-targ-UCU} with $p=1$ and $q=n-1$.

\begin{lemma}[\cite{sun2021asymptotically}, Lemma 12]\label{lem:UCG}
Given $m$ ancillary qubits, any $n$-qubit UCG $V^{[n-1]}_{\{n\}}$ can be implemented by a standard quantum circuit of size $O(2^n)$ and depth $O\left(n+\frac{2^n}{m+n}\right)$.
\end{lemma}

The following framework of a QSP circuit was given in \cite{grover2002creating,kerenidis2017quantum}.
\begin{lemma}\label{lem:QSP_framework_UCG}
The QSP problem can be solved by $n$ UCGs of growing sizes, $V^{[n-1]}_{\{n\}} \cdots V^{[2]}_{\{3\}}V^{[1]}_{\{2\}}V^{\emptyset}_{\{1\}}$.
\end{lemma}

\paragraph{Decomposition of $n$-qubit quantum gate}
Based on the cosine-sine decomposition, any $n$-qubit unitary can be decomposed into two $(1,n-1)$-UCUs and one $(n-1)$-UCG.

\begin{lemma} [\cite{paige1994history}]\label{lem:CSD}
Any $n$-qubit unitary $U\in\mathbb{C}^{2^n\times 2^n}$ can be decomposed as
\begin{equation}\label{eq:CSdecomp}
    U=\left(\begin{array}{cc}
   V_1'  & \\
     & V_1''
\end{array}\right)\left(\begin{array}{cc}
   C  & S\\
   -S  & C
\end{array}\right)\left(\begin{array}{cc}
   V_2'  & \\
     & V_2''
\end{array}\right),
\end{equation}
where $V_1',V_1'',V_2',V_2''\in\mathbb{C}^{2^{n-1}\times 2^{n-1}}$ are unitary matrices, $C,S\in \mathbb{C}^{2^{n-1}\times 2^{n-1}}$ are diagonal matrices whose diagonal elements are $\cos\theta_1$, $\cos\theta_2$, $\ldots$, $\cos\theta_{2^{n-1}}$ and $\sin\theta_1$, $\sin\theta_2$, $\ldots$, $\sin\theta_{2^{n-1}}$, respectively.
\end{lemma}
The circuit representation of the cosine-sine decomposition is shown in Figure \ref{fig:cosine-sine}.

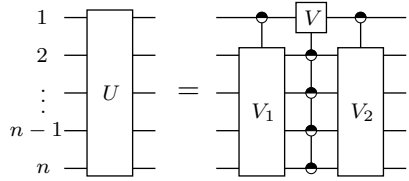
\begin{figure}[!ht]
    \centering
    \begin{tikzpicture}
    \draw[line width =0.5pt] (0,0) -- (1.2,0);
     \draw[line width =0.5pt] (0,-0.5) -- (1.2,-0.5);
     \draw[line width =0.5pt] (0,-1) -- (1.2,-1);
    \draw[line width =0.5pt] (0,-1.5) -- (1.2,-1.5);
    \draw[line width =0.5pt] (0,-2) -- (1.2,-2);
    \draw[line width =0.5pt,fill=white] (0.3,0.1) -- (0.9,0.1) -- (0.9,-2.1) -- (0.3, -2.1) -- cycle;
    \draw (0.37,-1) node[anchor=west]{\scriptsize $U$};

    \draw[line width =0.5pt] (2,0) -- (4.5,0);
     \draw[line width =0.5pt] (2,-0.5) -- (4.5,-0.5);
     \draw[line width =0.5pt] (2,-1) -- (4.5,-1);
    \draw[line width =0.5pt] (2,-1.5) -- (4.5,-1.5);
    \draw[line width =0.5pt] (2,-2) -- (4.5,-2);
    
    \draw [line width =0.5pt] (2.6,0) -- (2.6,-0.5);
    \draw [fill=white] (2.6,0) circle (0.07);
    \draw [fill=black] (2.53,0) arc (180:0:0.07) -- cycle;

    \draw [line width =0.5pt] (3.9,0) -- (3.9,-0.5);
    \draw [fill=white] (3.9,0) circle (0.07);
    \draw [fill=black] (3.83,0) arc (180:0:0.07) -- cycle;
    
    \draw [line width =0.5pt] (3.25,0) -- (3.25,-2);
    \draw [fill=white] (3.25,-0.5) circle (0.07);
    \draw [fill=black] (3.18,-0.5) arc (180:0:0.07) -- cycle;
    \draw [fill=white] (3.25,-1) circle (0.07);
    \draw [fill=black] (3.18,-1) arc (180:0:0.07) -- cycle;
    \draw [fill=white] (3.25,-1.5) circle (0.07);
    \draw [fill=black] (3.18,-1.5) arc (180:0:0.07) -- cycle;
    \draw [fill=white] (3.25,-2) circle (0.07);
    \draw [fill=black] (3.18,-2) arc (180:0:0.07) -- cycle;
    
    \draw[line width =0.5pt,fill=white] (2.3,-0.4) -- (2.9,-0.4) -- (2.9,-2.1) -- (2.3, -2.1) -- cycle;
    \draw[line width =0.5pt,fill=white] (3.05, 0.2) -- (3.45,0.2) -- (3.45,-0.2) -- (3.05,-0.2) -- cycle;
    \draw[line width =0.5pt,fill=white] (3.6,-0.4) -- (4.2,-0.4) -- (4.2,-2.1) -- (3.6, -2.1) -- cycle;
    
    \draw (2.33,-1.25) node[anchor=west]{\scriptsize $V_1$};
    \draw (3.63,-1.25) node[anchor=west]{\scriptsize $V_2$};
    \draw (3.02, 0) node[anchor=west]{\scriptsize $V$};
    \draw (1.35, -1) node[anchor=west]{=};
    
    \draw (-0.5, 0) node[anchor=west]{\scriptsize $1$};
    \draw (-0.5, -0.5) node[anchor=west]{\scriptsize $2$};
    \draw (-0.45, -1) node[anchor=west]{\scriptsize $\vdots$};
    \draw (-0.85, -1.5) node[anchor=west]{\scriptsize $n-1$};
    \draw (-0.5, -2) node[anchor=west]{\scriptsize $n$};
    \end{tikzpicture}
        \caption{The cosine-sine decomposition of an $n$-qubit quantum gate in the language of UCU.}
    \label{fig:cosine-sine}
\end{figure}

\section{Asymptotically optimal circuit depth for (controlled) quantum state preparation}
\label{sec:QSP_optimal}

Now we give a more detailed implementation and analyze the correctness and cost of the quantum circuit. Recall that we are constructing quantum circuits to implement QSP and CQSP, without assuming any QRAM hardware available.

First, we will use the following copying circuit many times so we single it out as a lemma.

\begin{lemma} [\cite{sun2021asymptotically}]\label{lem:copy1}
For any $x=x_1x_2\ldots x_n\in\{0,1\}^n$, a unitary transformation $U_{copy}$ satisfying
\[\ket{x}\ket{0^{mn}}\xrightarrow{U_{copy}} \ket{x}\underbrace{\ket{x}\ket{x}\cdots \ket{x}}_{m~\text{copies of}~\ket{x}},\]
can be implemented by a CNOT circuit of depth $O(\log m)$ and size $O(mn)$.
\end{lemma}

The next lemma says that a special type of UCU can be implemented efficiently.
\begin{lemma}\label{lem:multi-targ-UCU}
For all $i\in[p]$ and $x\in\{0,1\}^q$, suppose that $U^x_i$ is a single-qubit gate and let $L^x=\bigotimes_{i=1}^p U^x_i$. Then for any $m\ge pq$, the unitary 
$\sum_{x\in\{0,1\}^q}|x\rangle\langle x|\otimes L^x $ can be implemented by a standard quantum circuit of depth $O\left(\log p+q+\frac{p2^q}{m}\right)$ and size $O(p2^q)$ with $m$ ancillary qubits.
\end{lemma}
\begin{proof}
For any $x\in\{0,1\}^q$ and $y=y_1\cdots y_p\in\{0,1\}^p$, unitary $\sum_{x\in\{0,1\}^q}|x\rangle\langle x|\otimes L^x $ can be realized as follows.
 \begin{align}
    & \ket{x}\ket{y}\ket{0^m}\nonumber\\
    =&\ket{x}\left(\bigotimes_{i=1}^p\ket{y_i}\ket{0^q}_{\texttt{R}_i}\right)\ket{0^{m-pq}} \nonumber\\
    \xrightarrow{U_{copy}} & \ket{x}\left(\bigotimes_{i=1}^p\ket{y_i}\ket{x}_{\texttt{R}_i}\right)\ket{0^{m-pq}} & (\text{depth } O(\log p), \text{size } O(pq), \text{~Lemma~}\ref{lem:copy1}) \label{eq:multi-UCU1}\\
    \xrightarrow{\bigotimes_{i=1}^p\sum_{x} U_i^x\otimes |x\rangle_{\texttt{R}_i}\langle x|_{\texttt{R}_i}} & \ket{x}\left(\bigotimes_{i=1}^p U_i^x\ket{y_i}\ket{x}_{\texttt{R}_i}\right)\ket{0^{m-pq}} & (\text{depth } O\left(q+\frac{p2^q}{m}\right),~ \text{size } O(p2^q)) \label{eq:multi-UCU2}\\
    \xrightarrow{U_{copy}^\dagger} & \ket{x}\left(\bigotimes_{i=1}^pU_i^x\ket{y_i}\ket{0^q}_{\texttt{R}_i}\right)\ket{0^{m-pq}} & (\text{depth } O(\log p),  \text{size } O(pq), \text{~Lemma~}\ref{lem:copy1}) \label{eq:multi-UCU3}\\
    = & \ket{x}L^x\ket{y}\ket{0^m} \nonumber
\end{align}   

The first $pq$ ancillary qubits are divided into $p$ registers, which are labelled as register $\texttt{R}_1,\texttt{R}_2,\ldots,\texttt{R}_p$. Based on Lemma \ref{lem:copy1}, we make $p$ copies of $\ket{x}$, using a quantum circuit of depth $O(\log p)$ and size $O(pq)$ in Eq. \eqref{eq:multi-UCU1}. For every $i\in[p]$, we apply a $q$-UCG $\sum_{x\in\{0,1\}^q} U_i^x\otimes \ket{x}_{\texttt{R}_i}\langle x|_{\texttt{R}_i}$ on $\ket{y_i}\ket{x}_{\texttt{R}_i}$. All these $q$-UCGs act on different qubits, so they can be implemented in parallel, each with $\frac{m-pq}{p}$ ancillary qubits. According to Lemma \ref{lem:UCG},  Eq. \eqref{eq:multi-UCU2} can be realized by a quantum circuit of depth $O\left(q+\frac{2^q}{
(m-pq)/p+q}\right) = O\left(q+\frac{p2^q}{m}\right)$ and size $p\cdot O(2^q)=O(p2^q)$. In Eq. \eqref{eq:multi-UCU3}, we restore register $\texttt{R}_1,\texttt{R}_2,\ldots,\texttt{R}_p$ by the inverse circuit of Eq. \eqref{eq:multi-UCU1}, of depth $O(\log p)$ and size $O(pq)$. The total depth is $2\cdot O(\log p)+O\left(q+\frac{p2^q}{m}\right)=O\left(\log p+q+\frac{p2^q}{m}\right)$ and the total size is $2\cdot  O\left(pq\right)+O\left(p2^q\right)=O(p2^q)$.
\end{proof}

We can compare this result to Theorem 2 in \cite{low2018trading}, which says that the unitary $\sum_{x\in\{0,1\}^q}\ket{x}\bra{x}\otimes L^x$ can be implemented by a quantum circuit of depth $O(q^2+\frac{p2^q}{m})$ using $m$ ancillary qubits. Apart from the difference between their depth bound and ours, the assumptions are also different. On one hand, our construction works for any  $m\ge pq$, while theirs needs $\sqrt{p2^q}\le m\le p2^q$. On the other hand, ours takes $m$ ``clean'' qubits of $\ket{0}$ (and restore them afterwards), while they can handle ``dirty'' qubits, i.e. those with unknown content before the circuit.

\subsection{Rosenthal's quantum state preparation framework}
\label{sec:rosenthal}

In \cite{rosenthal2021query} Rosenthal presents a $\textsf{QAC}^0_f$ circuit of depth $O(n)$ with $O(n 2^n)$ ancillary qubits for $n$-qubit QSP. {As mentioned in \cite{rosenthal2021query}, this result suffices to yield a standard quantum circuit for QSP, with depth $O(n)$ and $O(n2^n)$ ancillary qubits. Indeed, each $k$-qubit Toffoli or fanout gate can be simulated by a standard quantum circuit of depth $O(\log k)$ with $O(k)$ ancillary qubits (Lemma \ref{lem:tof}). However, the $\textsf{QAC}^0_f$ circuit needs $O(n2^n)$ ancillary qubits, which is out of our parameter regime of $m \in \left[\omega\left(\frac{2^n}{n \log n}\right), o(2^n)\right]$.}

Next we will analyze the $\textsf{QAC}^0_f$ circuit and see how to make it suitable for any $m=\Omega(2^n/n^2)$. Let us first review Rosenthal's framework.
In the following, we will use $\epsilon$ to denote the empty string. Let $\{0,1\}^{\le n}\defeq\bigcup_{i=1}^n\{0,1\}^i \cup \{\epsilon\}$ denote the set of $\{0,1\}$ strings of length at most $n$, and $\{0,1\}^{< n}\defeq\bigcup_{i=1}^{n-1}\{0,1\}^i \cup \{\epsilon\}$ denote the set of $\{0,1\}$ strings of length at most $n-1$. For any $x=x_1x_2\cdots x_n\in\{0,1\}^n$, let $x_{\le i}$ denote the $i$-bit string $x_1x_2\cdots x_i$ and $x_{<i}$ denote the $(i-1)$-bit string $x_1x_2\cdots x_{i-1}$. {In particular, $x_{<1} = \epsilon$, the empty string.} Let $R(\alpha)$ denote a single-qubit gate $R(\alpha)=\left[\begin{array}{cc}
    1 &  \\
     & e^{i\alpha}
\end{array}\right]$ for any $\alpha\in\mathbb{R}$, which puts a phase of $\alpha$ on $\ket{1}$ basis.

Let $\ket{\psi_v}=\sum_{x\in\{0,1\}^n}v_x\ket{x}$ denote the target quantum state. For all $x\in\{0,1\}^{<n}$, let $|x|$ denote the length of $x$. Define $(n-|x|)$-qubit states $\{\ket{\psi_x}: 0 \le |x| < n\}$ recursively by the equations
\[\ket{\psi_\epsilon} = \ket{\psi_v} \ \text{ and }\ \ket{\psi_x}=\left\{\begin{array}{ll}
     \beta_{x0}\ket{0}\ket{\psi_{x0}}+\beta_{x1}\ket{1}\ket{\psi_{x1}},&\text{if }|x|\le n-2,   \\
    \beta_{x0}\ket{0}+\beta_{x1}\ket{1}, & \text{if }|x|=n-1, 
\end{array}\right.\]
It can be verified that $v_x=\prod_{i=1}^n\beta_{x_{\le i}}$ 
for all $x\in\{0,1\}^n$. For all $x\in\{0,1\}^{<n}$, further define a one-qubit quantum state
\begin{equation}\label{eq:phix}
    \ket{\phi_x}=\beta_{x0}\ket{0}+\beta_{x1}\ket{1}.
\end{equation}

Next let us define a \textit{leaf function} $\ell:\{0,1\}^{\{0,1\}^{<n}}\to \{0,1\}^n$. The set $\{0,1\}^{\{0,1\}^{<n}}$ consists of all bit strings of length $|{\{0,1\}^{<n}}|=2^n-1$. Each string in $\{0,1\}^{\{0,1\}^{<n}}$ has its bits indexed by elements in ${\{0,1\}^{<n}}$. For example, $\{0,1\}^{\{0,1\}^{<3}}$ consists of $7$-bit strings $z:=z_{\epsilon}z_0z_1z_{00}z_{01}z_{10}z_{11}$, where $z_x\in\{0,1\}$ for all $x\in {\{0,1\}^{<3}}$. The leaf function is defined in the following way: Identify the input index set $\{0,1\}^{<n}$ with the {interior} vertices of the complete binary tree {of height $n$}, with each interior vertex $x$ having the left and right children $x0$ and $x1$, respectively. The root corresponds to the empty string $\epsilon$. Given an input $z$, $\ell(z)$ is the leaf that the following walk from the root leads to: at any interior node $x$, move to the left or right child if $z_x = 0$ or $1$, respectively. For example, given an input $z:=z_{\epsilon}z_0z_1z_{00}z_{01}z_{10}z_{11}$=0110010, $\ell(z)$ is obtained as follows. First, since $z_\epsilon=0$, we move to the left child of $\epsilon$, which is labeled by $0$. Second, since $z_0=1$, we move to the right child of $0$, which is labeled by $01$. {Third, since $z_{01}=0$, we move to the left child of $01$, which is labeled by $010$ and} is the leaf node $\ell(z)$, i.e. $\ell(z)=01{0}$. It can be verified that
\[\ell(z)_j=\bigvee_{t\in \{0,1\}^j:t_j=1}\bigwedge_{i\in[j]}[z_{t_{<i}}=t_i],\forall j\in[n], \]
where
\[[z_{t_{<i}}=t_i]=\begin{cases} 1 & \text{ if } z_{t_{<i}}=t_i, \\ 0 & \text{ if }z_{t_{<i}}\neq t_i \end{cases}.\] 
Also define a corresponding $(2^n+n-1)$-qubit unitary transformation $U_\ell$ by
\begin{equation}\label{eq:U_ell}
    U_{\ell}\ket{z,a}=\ket{z,a\oplus \ell(z)},\quad \forall z\in\{0,1\}^{\{0,1\}^{<n}},\forall a\in\{0,1\}^n.
\end{equation}
In the rest of this section, $\texttt{R}_x$ is a one-qubit register for each $x\in\{0,1\}^{<n}$ 
, and $\texttt{S}$ is an $n$-qubit register.

The QSP algorithm in \cite{rosenthal2021query} can be summarized as follows.
\begin{lemma}\label{lem:rosenthal-alg}
    Any $n$-qubit quantum state $\ket{\psi}$ can be generated by the following three steps:
    \begin{enumerate}
        \item $\ket{0}_{{\tt R}_x} \to \ket{\phi_x}_{{\tt R}_x}$, for all $x\in \{0,1\}^{< n}$.
        \item Apply $U_\ell$ to $\bigotimes_{x\in \{0,1\}^{< n}}  \ket{\phi_x}_{{\tt R}_x}\otimes \ket{0^n}_{{\tt S}}$.
        \item Apply $\Gamma^\dagger$, where $\Gamma$ is any unitary satisfying
        \begin{align}\label{eq:Gamma}
            &\ket{t}_{{\tt S}}\bigotimes_{x\in\{0,1\}^{<n}} \ket{0}_{{\tt R}_x}  
            \xrightarrow{\Gamma}
            \ket{t}_{\tt{S}}\bigotimes_{x\in\{0,1\}^{<n}}\left\{
            \begin{array}{ll}
                \ket{t_i}_{{\tt R}_x}&  \text{if } x=t_{<i} \text{ for some } i\in[n],\\
                \ket{\phi_x}_{{\tt R}_x},  & \text{otherwise},
            \end{array}\right. \forall t\in\{0,1\}^n.
        \end{align} 
    \end{enumerate}
\end{lemma}
For correctness please refer to \cite{rosenthal2021query}, and here we focus on the implementation and the corresponding analysis in a way suitable for our later circuit construction. 

The first step of the algorithm consists of single-qubit rotations on $2^n-1$ qubits, and thus naturally has depth 1 and size $2^n-1$. We denote by $L_1^v$ this step of operation, where the superscript emphasizes that the gate parameters depend on the target vector $v\in \mathbb{C}^{2^n}$. 

As shown in \cite{rosenthal2021query}, the second step can be implemented by a $\textsf{QAC}_f^0$ circuit on $O(n2^n)$ qubits, which transfers to a standard circuit of depth $O(n)$ and size $O(n2^n)$, with $O(n2^n)$ ancillary qubits. We also note that this second step is independent of the target state.  We denote by $C_1'$ the circuit of this step, where the absence of superscript $v$ emphasizes the independence of the target state. 

The third step, though also of depth $O(n)$ and size $O(n2^n)$ with $O(n2^n)$ ancillary qubits, unfortunately, depends on the target vector $v$. This brings us some difficulty for small $m$, and we will show how to handle it next.

\subsection{Implementation: separating depth and dependence on the target state}
In this section we will show how to implement the third step in Rosenthal's algorithm in such a way that (1) it has a constant number of rounds, some deep and some shallow, (2) deep rounds have depth $O(n)$, but are independent of the target vector $v$, (3) shallow rounds each have depth $1$, and depend on $v$. This separation of depth and dependence is useful for our later construction of efficient circuits.
The circuit and these conditions are formalized in the following lemma. 
\begin{lemma}\label{lem:app_controlled}
A unitary transformation $\Gamma^\dagger$ satisfying Eq.\eqref{eq:Gamma} can be implemented by a standard quantum circuit of the following form 
\[\Gamma^\dagger = C_5 L_5^v C_4 L_4^v C_3 L_3^v C_2 L_2^v C_1''.\]
Here each $L_i^v=\bigotimes_{k=1}^{s_i} U^{i,v}_k$ is a depth-1 circuit consisting of $s_i=O(2^n)$ single-qubit gates with $U^{i,v}_k$ determined by $\ket{\psi_v}$. $C_1''$, $C_2, \ldots, C_5$ are all independent of $\ket{\psi_v}$; $C_1''$, $C_4$ and $C_5$ are circuits of depth $O(n)$ and size $O(n2^n)$, and $C_2$ and $C_3$ are circuits of depth $O(1)$ and size $O(2^n)$. {The circuit uses  $O(n2^n)$ ancillary qubits.}
\end{lemma}
\begin{proof}
We first introduce notation $C^{\texttt{S},y}_{\texttt{R}_x}(V)$: For any $y=y_1\cdots y_\ell \in\{0,1\}^{\le n}$, $C^{\texttt{S},y}_{\texttt{R}_x}(V)$ is a unitary operation acting on an $n$-qubit register $\texttt{S}$ (the first $n$ qubits) and 1-qubit register $\texttt{R}_x$ (the last qubit) as follows \[C^{\texttt{S},y}_{\texttt{R}_x}(V)\defeq \ket{y}\langle y|\otimes \mathbb{I}_{n-\ell}\otimes V + \sum_{y'\in\{0,1\}^\ell-\{y\}}\ket{y'}\langle y'|\otimes \mathbb{I}_{n-\ell}\otimes \mathbb{I}_1,\]
The unitary $C^{\texttt{S},y}_{\texttt{R}_x}(V)$ makes the following transformation
\begin{equation}\label{eq:multi-control}
\ket{t}_{\texttt{S}}\ket{0}_{\texttt{R}_x}\to\ket{t}_{\texttt{S}}V^{[t_{\le \ell }=y]}\ket{0}_{\texttt{R}_x},
\end{equation}
where
\[ [t_{\le\ell}=y]\defeq \left\{\begin{array}{ll}
    1, & \text{if~} t_{\le\ell}=y, \\
    0, & \text{if~} t_{\le\ell}\neq y.
\end{array}\right.\]

By introducing an ancillary qubit called register $\texttt{A}$, $C^{\texttt{S},y}_{\texttt{R}_x}(V)$ can be implemented by the quantum circuit in Figure \ref{fig:C^S_Rx}. In the quantum circuit, the single-qubit gate $A,B,C,R(\alpha)$ satisfy $V=e^{i\alpha}AXBXC$ and $ABC=\mathbb{I}_1$ (Lemma \ref{lem:1Qfactor}). Indeed, for the first $\ell$ qubits being $y$, register {\texttt A} becomes $\ket{1}$, which triggers the later application of $AXBXC$ to $\texttt{R}_x$ followed by a global phase $e^{i\alpha}$. Registers $\texttt{R}_x$ and $\texttt{A}$ are then restored. If the first $\ell$ qubits are not $y$, then register $\texttt{A}$ remains $\ket{0}$ all the way, and on register $\texttt{R}_x$ is applied $ABC$, which also equals to $I$.

\begin{figure}[!ht]
    \centering
\centerline{
\Qcircuit @C=1.2em @R=0.5em {
&\gate{\scriptstyle X^{y_1}} & \ctrl{1} & \qw & \qw & \qw & \qw & \qw & \ctrl{1} & \gate{\scriptstyle X^{y_1}} & \qw &\\
&\gate{\scriptstyle X^{y_2}} & \ctrl{1} & \qw &\qw & \qw & \qw & \qw &  \ctrl{1} &\gate{\scriptstyle X^{y_2}} &  \qw &\\
 &  & \vdots &  & & \cdots &   &   & \vdots &   & &\\
&\gate{\scriptstyle X^{y_\ell}} &\ctrl{5} &\qw &\qw & \qw & \qw & \qw &  \ctrl{5} & \gate{\scriptstyle X^{y_\ell}} & \qw &\\
& \qw & \qw & \qw &\qw & \qw & \qw & \qw & \qw & \qw & \qw &\\
 &  & & & & & & & & &\\
& \qw & \qw & \qw &\qw & \qw & \qw & \qw & \qw & \qw &\qw \inputgroupv{1}{6}{2em}{3 em}{ \scriptstyle \text{Register}~~\texttt{S}~~~~~~~~~~~~}\\
\lstick{\scriptsize \text{Register~} \texttt{R}_x~~~~~}& \qw &\qw & \gate{\scriptstyle C} & \targ &\gate{\scriptstyle B} &\targ & \gate{\scriptstyle A} &\qw & \qw & \qw &\\
\lstick{\scriptstyle \text{Register~} \texttt{A}~\ket{0}~ }& \qw & \targ &\qw& \ctrl{-1} &\qw & \ctrl{-1} & \gate{\scriptstyle R(\alpha)}&\targ &\qw & \qw & \rstick{\scriptstyle \ket{0}}\\
}
}
    \caption{Quantum circuit for $C^{\texttt{S},y}_{\texttt{R}_x}(V)$. Register $\texttt{A}$ is an ancillary qubit. Single-qubit gate $A,B,C,R(\alpha)$ satisfy $V=e^{i\alpha}AXBXC$ and $ABC=\mathbb{I}_1$.}
    \label{fig:C^S_Rx}
\end{figure}
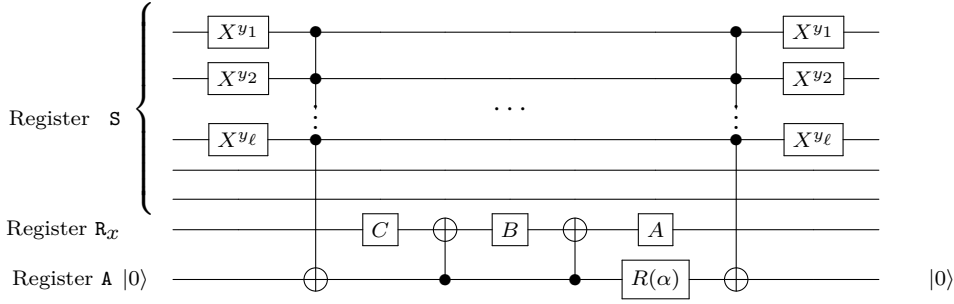

According to Figure \ref{fig:C^S_Rx}, we can rewrite the circuit of $C^{\texttt{S},y}_{\texttt{R}_x}(V)$:
\begin{multline}\label{eq:C^S_Rx}
C^{\texttt{S},y}_{\texttt{R}_x}(V)= W_y^2\underbrace{(\mathbb{I}_n\otimes A\otimes R(\alpha))}_{D^3_{V}}(\mathbb{I}_n\otimes CNOT^{\texttt{A}}_{\texttt{R}_x})\underbrace{(\mathbb{I}_n\otimes B\otimes \mathbb{I}_1)}_{D^2_{V}}(\mathbb{I}_n\otimes CNOT^{\texttt{A}}_{\texttt{R}_x}) \underbrace{(\mathbb{I}_n\otimes C\otimes \mathbb{I}_1)}_{D^1_{V}}W_y^1,
\end{multline}
where $W_y^1,W_y^2$ are defined as
\begin{align*}
    & W_y^1:={\sf Tof}^{\texttt{S}}_\texttt{A}(\bigotimes_{i=1}^{y_\ell}X^{y_i}),\\
    & W_y^2:=(\bigotimes_{i=1}^{y_\ell}X^{y_i}){\sf Tof}^{\texttt{S}}_\texttt{A},
\end{align*}
and ${\sf Tof}^{\texttt{S}}_\texttt{A}$ is a Toffoli gate whose control qubits are in $\texttt{S}$ and target qubit is in $\texttt{A}$.
Because any $n$-qubit Toffoli gate can be implemented by a quantum circuit of depth $O(n)$ based on Lemma \ref{lem:tof}, $W_y^1,W_y^2$ can be implemented by a quantum circuit of depth  $O(n)$. Unitary $D_v^1,D_V^2$ consists of single-qubit gates and $D_V^3$ consists of $2$ single-qubit gates. The total depth of $C^{\texttt{S},y}_{\texttt{R}_x}(V)$ is $O(n)$. 
It is worth mentioning that in the circuit construction of $C^{\texttt{S},y}_{\texttt{R}_x}(V)$, only $D_V^1,D_V^2, D_V^3$ depend on unitary $V$. 

Now we start the circuit construction of $\Gamma$. For all $x\in\{0,1\}^{<n}$, let $U_x$ denote a single-qubit gate satisfying $U_x\ket{0}=\ket{\phi_x}$.
First, we implement the following transformation $\Gamma_{\texttt{R}_x}^{\texttt{S},x}$ on register $\texttt{S}$ and $\texttt{R}_x$ by the method in Figure \ref{fig:C^S_Rx}:
\begin{equation}\label{eq:a}
 \ket{t}_{\texttt{S}}\ket{0}_{\texttt{R}_x}\to \ket{t}_{\texttt{S}}\begin{cases}
   \ket{t_i}_{\texttt{R}_x}, &  \text{if~} x=t_{<i}\text{~for~some~} i\in[n],\\
  \ket{\phi_x}_{\texttt{R}_x},  & \text{otherwise,}
\end{cases}
\qquad \forall t\in \{0,1\}^n.   
\end{equation}
This needs depth and size $O(n)$ with 1 ancillary qubit:
\begin{align}
    & \ket{t}_{\texttt{S}}\ket{0}_{\texttt{R}_x}\ket{0}_{\texttt{A}} \nonumber\\
    \xrightarrow{\mathbb{I}_n\otimes U_x} & \ket{t}_{\texttt{S}}\ket{\phi_x}_{\texttt{R}_x}\ket{0}_{\texttt{A}} & (\text{depth~} 1, \text{~size~} 1) \label{eq:Gamma_x1}\\
    \xrightarrow{C^{\texttt{S},x}_{\texttt{R}_x}(U_x^\dagger)} & \ket{t}_{\texttt{S}}\begin{cases}
   \ket{0}_{\texttt{R}_x}, &  \text{if~} x=t_{<i}\text{~for~some~} i\in[n]\\
  \ket{\phi_x}_{\texttt{R}_x},  & \text{otherwise}
  \end{cases}
  \ket{0}_{\texttt{A}} & (\text{depth~} O(n), \text{~size~} O(n)) \label{eq:Gamma_x2}\\
    \xrightarrow{C^{\texttt{S},x1}_{\texttt{R}_x}(X)} & \ket{t}_{\texttt{S}}\begin{cases}
   \ket{t_i}_{\texttt{R}_x}, &  \text{if~} x=t_{<i}\text{~for~some~} i\in[n]\\
  \ket{\phi_x}_{\texttt{R}_x},  & \text{otherwise}
\end{cases}   \ket{0}_{\texttt{A}} & (\text{depth~} O(n),\text{~size~} O(n)) \label{eq:Gamma_x3}
\end{align}

For each $x\in\{0,1\}^{<n}$, $\texttt{S}_x$ denotes a register which stores a copy of $\ket{t}$ in register $\texttt{S}$. Base on the construction of $\Gamma_{\texttt{R}_x}^{\texttt{S}_x,x}$, we can now implement $\Gamma$ by a quantum circuit of depth $O(n)$ and size $O(n2^n)$ with $(n+1)(2^n-1)$ ancillary qubits.
To compress the depth, we first make a copy of $\ket{t}$ for each $x\in\{0,1\}^{<n}$.
\begin{align}
    &\ket{t}_{\texttt{S}}\left(\bigotimes_{x\in\{0,1\}^{<n}}\ket{0}_{\texttt{R}_x}\right)\ket{0^{(n+1)(2^n-1)}} \nonumber\\
    =&\ket{t}_{\texttt{S}}\bigotimes_{x\in\{0,1\}^{<n}}\left(\ket{0}_{\texttt{R}_x}\ket{0^n}_{\texttt{S}_x}\ket{0}_{\texttt{A}_x}\right) \label{eq:a1}\\
    \xrightarrow{U_{copy}}& \ket{t}_{\texttt{S}}\bigotimes_{x\in\{0,1\}^{<n}}\left(\ket{0}_{\texttt{R}_x}\ket{t}_{\texttt{S}_x}\ket{0}_{\texttt{A}_x}\right) &  \label{eq:a2}\\
    \xrightarrow{\bigotimes \limits_{x\in\{0,1\}^{<n}} \Gamma_{\texttt{R}_x}^{\texttt{S}_x,x}}& \ket{t}_{\texttt{S}}\bigotimes_{x\in\{0,1\}^{<n}}\left(\left\{\begin{array}{ll}
   \ket{t_i}_{\texttt{R}_x},&  \text{if~} x=t_{<i}\text{~for~some~} i\in[n],\\
  \ket{\phi_x}_{\texttt{R}_x},  & \text{otherwise}.
    \end{array}\right.\ket{t}_{\texttt{S}_x}\ket{0}_{\texttt{A}_x}\right) & \label{eq:a3}\\
    \xrightarrow{U^\dagger_{copy}}& \ket{t}_{\texttt{S}}\bigotimes_{x\in\{0,1\}^{<n}}\left(\left\{\begin{array}{ll}
   \ket{t_i}_{\texttt{R}_x},&  \text{if~} x=t_{<i}\text{~for~some~} i\in[n],\\
   \ket{\phi_x}_{\texttt{R}_x},  & \text{otherwise}.
    \end{array}\right.\ket{0^n}_{\texttt{S}_x}\ket{0}_{\texttt{A}_x}\right)& \label{eq:a4}\\
    = & \ket{t}_{\texttt{S}} \left(\bigotimes_{x\in\{0,1\}^{<n}}
    \begin{cases} \ket{t_i}_{\texttt{R}_x},&  \text{if~} x=t_{<i}\text{~for~some~} i\in[n],\\
  \ket{\phi_x}_{\texttt{R}_x},  & \text{otherwise}.
    \end{cases} \right)\ket{0^{(n+1)(2^n-1)}} .\nonumber
\end{align}
Here all the three transformation steps have depth $O(n)$ and size $O(n2^n)$, by Lemma \ref{lem:copy1} and the analysis in Eq.\eqref{eq:Gamma_x1}-\eqref{eq:Gamma_x3}.

For every $U^\dagger_x$, $C^{S_x,x}_{\texttt{R}_x}(U_x^\dagger)$ can be represented as
\begin{equation}\label{eq:CSx}
  C^{S_x,x}_{\texttt{R}_x}(U_x^\dagger)=W_x^2 D_{U_x^\dagger}^3CNOT^{\texttt{A}_x}_{\texttt{R}_x} D_{U_x^\dagger}^2 CNOT^{\texttt{A}_x}_{\texttt{R}_x}D_{U_x^\dagger}^1W_x^1,  
\end{equation}
as discussed in Eq. \eqref{eq:C^S_Rx}. $D_{U_x^\dagger}^1,D_{U_x^\dagger}^2$ are single-qubit gate respectively and $D_{U_x^\dagger}^3$ consists of 2 single-qubit gates which are determined by $U_x^\dagger$ (or by target quantum state $\ket{\psi_v}$). $W_x^1,W_x^2$ are quantum circuits of depth $O(n)$, independent of $\ket{\psi_v}$.
As discussed above, $\Gamma$ is represented as
 \begin{align*}
\Gamma&=U_{copy}^\dagger \left(\bigotimes_{x\in\{0,1\}^{<n}}\Gamma^{S_x,x}_{\texttt{R}_x}\right) U_{copy},\\
     & =U_{copy}^\dagger \bigotimes_{x\in\{0,1\}^{<n}}(C^{\texttt{S}_x,x1}_{\texttt{R}_x}(X) W_x^2 D_{U_x^\dagger}^3CNOT^{\texttt{A}_x}_{\texttt{R}_x} D_{U_x^\dagger}^2 CNOT^{\texttt{A}_x}_{\texttt{R}_x}D_{U_x^\dagger}^1W_x^1 U_x) U_{copy},\\
     & =\underbrace{U_{copy}^\dagger \bigotimes_{x\in\{0,1\}^{<n}}(C^{\texttt{S}_x,x1}_{\texttt{R}_x}(X) W_x^2)}_{(C_1'')^\dagger} \underbrace{\bigotimes_{x\in\{0,1\}^{<n}}D_{U_x^\dagger}^3}_{(L_2^{v})^\dagger}\underbrace{\bigotimes_{x\in\{0,1\}^{<n}} CNOT^{\texttt{A}_x}_{\texttt{R}_x}}_{(C_2)^\dagger}\\ & \underbrace{\bigotimes_{x\in\{0,1\}^{<n}}D_{U_x^\dagger}^2}_{(L_3^v)^\dagger}  \underbrace{\bigotimes_{x\in\{0,1\}^{<n}} CNOT^{\texttt{A}_x}_{\texttt{R}_x}}_{(C_3)^\dagger} \underbrace{\bigotimes_{x\in\{0,1\}^{<n}}D_{U_x^\dagger}^1 }_{(L_4^v)^\dagger}\underbrace{\bigotimes_{x\in\{0,1\}^{<n}} W_x^1}_{(C_4)^\dagger} \underbrace{\bigotimes_{x\in\{0,1\}^{<n}} U_x}_{(L_5^v)^\dagger} \underbrace{U_{copy}}_{(C_5)^\dagger}.
\end{align*}   

The conclusion for the decomposition of $\Gamma^\dagger$ then follows. For the cost analysis: According to Eq. \eqref{eq:CSx}, $(L^v_2)^\dagger,(L^v_3)^\dagger,(L^v_4)^\dagger,(L^v_5)^\dagger$ are depth-1 circuits consisting of $O(2^n)$ single-qubit gates, which are determined by target state $\ket{\psi_v}$. 
According to Lemma \ref{lem:copy1}, Figure \ref{fig:C^S_Rx}, Eq. \eqref{eq:Gamma_x3} and \eqref{eq:CSx}, $(C_1'')^\dagger$, $(C_4)^\dagger$ and $(C_5)^\dagger$ are quantum circuits of depth $O(n)$ and size $O(n2^n)$. 
Based on Eq. \eqref{eq:CSx}, $(C_2)^\dagger,(C_3)^\dagger$ are quantum circuits of depth $O(1)$ and size $O(2^n)$. This completes the proof.
\end{proof}

Recall that {$L_1^v$ and} $C_1'$ {are the first and} the second step in Lemma \ref{lem:rosenthal-alg}, respectively. Now letting $C_1 = C_1''C_1'$, we get the following result.
\begin{lemma}\label{lem:rosen}
Any $n$-qubit quantum state $\ket{\psi_v}$ can be generated by a quantum circuit $QSP_v$, using single-qubit gates and CNOT gates, of depth $O(n)$ and size $O(n2^n)$, with $O(n2^n)$ ancillary qubits. The QSP circuit can be written as
\[QSP_v=C_5L_5^vC_4L_4^vC_3L_3^vC_2L_2^vC_1L_1^v.\]
Each $L_i^v=\bigotimes_{k=1}^{s_i} U^{i,v}_k$ is a depth-1 circuit consisting of $s_i=O(2^n)$ single-qubit gates, and $L_i^v$ is determined by $\ket{\psi_v}$. $C_1,C_4$ and $C_5$ are circuits of depth $O(n)$ and size $O(n2^n)$, and $C_2$ and $C_3$ are circuits of depth $O(1)$ and size $O(2^n)$. For any $i\in[5]$, $C_i$ is independent of $\ket{\psi_v}$. 
\end{lemma}

%

\subsection{Quantum circuit for (controlled) quantum state preparation}
Next, we will use Lemma \ref{lem:rosen} to efficiently realize the controlled quantum state preparation. Let us fix a constant $c$ in the size upper bound of $L_i^v$ in Lemma \ref{lem:rosen}, i.e. $s_i \le c\cdot 2^n$. The next lemma is a restatement of the upper bound part of Theorem \ref{thm:multi-QSP}.



\begin{lemma}
\label{lem:multi-QSP}
For any $k\ge 0$ and quantum states $\{\ket{\psi_i}:i\in \{0,1\}^k\}$, the following controlled quantum state preparation
\[\ket{i}\ket{0^{n}}\to\ket{i}\ket{\psi_i}, \ \forall i\in \{0,1\}^k,\]
can be implemented by a standard quantum circuit of depth $O\left(n+k+\frac{2^{n+k}}{n+k+m}\right)$ and size $O\left(2^{n+k}\right)$ with $m$ ancillary qubits.
\end{lemma}
\begin{proof}
We consider two cases depending on $m$. 

\paragraph{Case 1: $m=O(2^{n+k}/(n+k)^2)$.} Let $QSP_i$ denote the QSP circuit for generating quantum state $\ket{\psi_i}$ on qubits $\{k+1,k+2,\cdots,k+n\}$ obtained from Lemma \ref{lem:QSP_framework_UCG}, then $QSP_i$ can be decomposed into $n$ UCGs:
\[
   QSP_i=V^{\{k+1,k+2,\cdots, n+k-1\}}_{\{n+k\}}(i)\cdots V^{\{k+1,k+2\}}_{\{k+3\}}(i)V^{\{k+1\}}_{\{k+2\}}(i) V^{\emptyset}_{\{k+1\}}(i). 
\]
Therefore, the controlled quantum state preparation can be implemented as
\begin{align*}
    & \sum_{i\in \{0,1\}^k}|i\rangle\langle i|\otimes QSP_i  \\
    =& \sum_{i\in \{0,1\}^k}|i\rangle\langle i|\otimes (V^{\{k+1,k+2,\cdots, n+k-1\}}_{\{n+k\}}(i)
    \cdots V^{\{k+1,k+2\}}_{\{k+3\}}(i) V^{\{k+1\}}_{\{k+2\}}(i) V^{\emptyset}_{\{k+1\}}(i))\\
    = & V^{[n+k-1]}_{\{n+k\}}\cdots V^{[k+2]}_{\{k+3\}} V^{[k+1]}_{\{k+2\}} V^{[k]}_{\{k+1\}}.
\end{align*}
For all $i\in[n]$, UCG $V^{[k+i-1]}_{\{k+i\}}$ can be implemented by a quantum circuit of depth $O\left(k+i+\frac{2^{k+i}}{k+i+m}\right)$ and size $O(2^{k+i})$ by Lemma \ref{lem:UCG}, using $m$ ancillary qubits. Therefore, the depth and size of this CQSP circuit are $\sum_{i=1}^{n}O\left(k+i+\frac{2^{k+i}}{k+i+m}\right)=O\left(\frac{2^{n+k}}{m+n+k}\right)$ and $\sum_{i=1}^{n}O\left(2^{k+i}\right)=O\left(2^{n+k}\right)$, respectively.

\paragraph{Case 2: $m=\Omega(2^{n+k}/(n+k)^2)$.} We will show the quantum circuit for CQSP in two sub-cases: $k\ge \lceil 4\log (n+k)\rceil$ and $k<\lceil 4\log (n+k)\rceil$. 

\paragraph{Case 2.1: $k\ge\lceil 4\log (n+k)\rceil$.} Then $m\ge \max\{2cn2^{n},k2^{n}\}$ for any constant $c>0$.
For all ${i\in \{0,1\}^k}$, suppose $\ket{\psi_i} =  \sum_{j=0}^{2^{n}-1}v^i_j\ket{j}$. Let $QSP_i$ denote a QSP circuit with $m_1=cn2^{n}$ ancillary qubits as guaranteed by Lemma \ref{lem:rosen} to prepare $\ket{\psi_i}$, which can be represented as 
\[QSP_i=C_5L_5^iC_4L_4^iC_3L_3^iC_2L_2^iC_1L_1^i.\]
Here each $L_r^i=\bigotimes_{j=0}^{s_r}U_j^{r,i}$ is a depth-1 circuit consisting of $s_r=O(2^{n})$ single-qubit gates, and $L_r^i$ is determined by $\ket{\psi_i}$. For $r\in [5]$, $C_r$ is an $(n+m_1)$-qubit circuit of depth $O(n)$, which is independent of $\ket{\psi_i}$. Note that the task in the statement of this lemma is nothing but the UCU of $\{QSP_i\}$, which can be implemented by applying $\sum_{i\in \{0,1\}^k}|i\rangle\langle i|\otimes QSP_i$. This operator can be decomposed as follows.
\begin{align}\label{eq:multi_qsp}
    & \sum_{i\in \{0,1\}^k}|i\rangle\langle i|\otimes QSP_i \nonumber \\
    =& \sum_{i\in \{0,1\}^k}|i\rangle\langle i|\otimes (C_5L_5^iC_4L_4^iC_3L_3^iC_2L_2^iC_1L_1^i)\\
    = & \prod_{r=5}^1\Big[(\mathbb{I}_t\otimes C_r)\Big(\sum_{i\in \{0,1\}^k}|i\rangle\langle i|\otimes L_r^i \Big)\Big].\nonumber
\end{align}
where the notation $\prod_{r=5}^1 A_r $ means to multiply the matrices $A_r$'s in the order of $A_5A_4A_3A_2A_1$. The second equation above holds because, when viewed as matrices, the equation is just a block diagonal matrix multiplication: 
\begin{align} \label{eq:switch}
    & \text{diag}(C_5L_5^0 C_4 L_4^0 C_3 L_3^0 C_2 L_2^0 C_1 L_1^0, \cdots, C_5L_5^{2^k-1} C_4 L_4^{2^k-1} C_3 L_3^{2^k-1} C_2 L_2^{2^k-1} C_1 L_1^{2^k-1}) \\
    = & \text{diag}(C_5, \cdots, C_5)\times \text{diag}(L_5^0, \cdots, L_5^{2^k-1}) \nonumber\times \text{diag}(C_4, \cdots, C_4) \times \cdots \times \text{diag}(L_1^0, \cdots, L_1^{2^k-1}).
\end{align}  

In the $m$ ancillary qubits, we used $m_1$ for the UCU $\{QSP_i\}$, and have $m-m_1$ left. Since $m\ge k2^{n}$, we can apply Lemma \ref{lem:multi-targ-UCU} (where $q=k$ and $p\le c 2^{n}$) and obtain that, for every $r\in[5]$, $\sum_{i\in \{0,1\}^k}|i\rangle\langle i|\otimes L_r^i $ can be implemented by a quantum circuit of depth 
\[O\left(n+k+\frac{2^{n}\cdot 2^k}{m-cn2^{n}}\right) = O\left(n+k+\frac{2^{n+k}}{m-cn2^{n}}\right),\]
and size $O\left(2^{n}\times 2^{k}\right)=O(2^{n+k})$, with the $m-m_1$ ancillary qubits. For $r\in[5]$, every $C_r$ is a quantum circuit of depth $O(n)$ and size $O\left(n2^{n}\right)$. Putting everything together and noting $m\ge 2cn2^{n}$, we can implement  $\sum_{i\in \{0,1\}^k}|i\rangle\langle i|\otimes QSP_i$ by a quantum circuit of depth $O\left(n+k+\frac{2^{n+k}}{m-cn2^{n}}\right){=O\left(n+k+\frac{2^{n+k}}{m}\right)}$ and size $O\left(2^{n+k} + n2^{n}\right){=O(2^{n+k})}$,  with $m$ ancillary qubits. 

\paragraph{Case 2.2 $k<\lceil 4\log (n+k)\rceil$.} Define $n$-qubit quantum state $\ket{\psi_i}=\sum_{\tau=0}^{2^{n}-1}v_{\tau,i}\ket{\tau}$ and $\ket{\psi_{i}^{(s)}} = \sum_{\eta=0}^{2^{s}-1} v'_{\eta,i} \ket{\eta}$, where $s\le n$ and $v'_{\eta,i}=\sqrt{\sum_{p=0}^{2^{n-s}-1}|v_{\eta\cdot 2^{n-s}+p,i}|^2}$. Our construction consists of two steps. In the first step, we implement a $\lceil 4\log (n+k)\rceil$-qubit CQSP, using $m=\Omega(2^{n+k}/(n+k)^2)$ ancillary qubits:
\begin{equation}
    \text{CQSP1}:~\ket{i}\ket{0^{n}}\to\ket{i}\ket{\psi_i^{(s)}}\ket{0^{n+k-\lceil 4\log (n+k)\rceil}}, \forall i\in \{0,1\}^k, \label{eq:CQSP1}
    \end{equation}
where $s = \lceil 4\log (n+k)\rceil-k$. Note that $k$ and $s$ satisfy $m>\max\{2cs2^{s},k2^{s}\}$, thus similar to Case 2.1 above, we can implement Eq. \eqref{eq:CQSP1} by a circuit of depth $O(\log (n+k))$ and size $O((n+k)^4)$. In the second step, we  implement an $(n+k)$-qubit CQSP using $m$ ancillary qubits:
\begin{equation}\label{eq:CQSP2}
   \text{CQSP2}:~ \ket{i}\ket{\eta}\ket{0^{(n+k)-\lceil 4\log (n+k)\rceil}}  \to \\ \ket{i}\ket{\eta}\ket{\phi_{i,\eta}}, \forall i\in \{0,1\}^k,\eta\in\{0\}\cup [2^s-1],
\end{equation}
where $\ket{\phi_{i,\eta}}\defeq\sum_{p=0}^{2^{(n+k)-\lceil 4\log (n+k)\rceil}}v_{\eta\cdot 2^{(n+k)-\lceil 4\log (n+k)\rceil}+p,i}/v'_{\eta,i}\ket{p}$. Eq. \eqref{eq:CQSP2} is a CQSP, in which the number of controlled qubits $\lceil 4\log (n+k)\rceil$ satisfying $m>\max\{2c((n+k)-\lceil 4\log (n+k)\rceil)2^{(n+k)-\lceil4\log (n+k)\rceil}, \lceil 4\log (n+k)\rceil2^{(n+k)-\lceil4\log (n+k)\rceil}\}$. Therefore Eq. \eqref{eq:CQSP2} can be implemented in the same way as in Case 2.1, such that the depth and size are $O\left(n+k+\frac{2^{n+k}}{n+k+m}\right)$ and $O(2^{n+k})$, respectively. It can be verified that the CQSP operator can be implemented by $\text{CQSP2}\cdot \text{ CQSP1}$ and the depth and size are $O\left(n+k+\frac{2^{n+k}}{n+k+m}\right)$ and $O(2^{n+k})$, respectively.
\end{proof}

The paper \cite{sun2021asymptotically} presents an $n$-qubit QSP circuit with $m$ ancillary qubits. For $m\in\left[0,O\left(\frac{2^n}{n\log n}\right)\right]\cup \left[\Omega(2^n),+\infty\right)$, the circuit depth for $n$-qubit quantum state preparation is optimal. However, if $m \in \left[\omega\left(\frac{2^n}{n\log n}\right), o(2^n)\right]$, there still exists a logarithmic gap between the upper and lower bounds of QSP circuit depth. Since our Theorem \ref{thm:multi-QSP} gives a unified construction that works for any $k$, including $k=0$, we obtain Theorem \ref{thm:QSP}, which closes the gap left open in \cite{sun2021asymptotically}. 
\paragraph{Remarks}
\begin{enumerate}
    \item {In \cite{plesch2011quantum}, it was shown that any $n$-qubit quantum states are determined by $2^n-1$ free parameters omitting a global phase. In Theorem \ref{thm:multi-QSP}, an $(n+k)$-qubit CQSP is defined by $2^{k}$ $n$-qubit quantum states. Therefore, it is determined by $2^{k}\cdot 2^{n}-1=2^{n+k}-1$ free parameters.
    Thus by a similar argument for the depth lower bound of the quantum state preparation in \cite{plesch2011quantum}, we can get a depth lower bound for $(k,n)$-qubit CQSP is $\Omega\left(\frac{2^{n+k}}{n+k+m}\right)$ using $m$ ancillary qubits.
    Moreover, the same as the proof of Lemma 37 in \cite{sun2021asymptotically}, we can also obtain a depth lower bound $\Omega(n+k)$ by the light cone argument. Combining the two results above, the depth lower bound for CQSP is $\Omega\left(n+k+\frac{2^{n+k}}{n+k+m}\right)$. Therefore, the circuit depth in Theorem \ref{thm:multi-QSP} is optimal.}
    \item If the number of controlled qubits $k$ in Theorem \ref{thm:multi-QSP} is $0$, the CQSP degenerates to standard QSP. Therefore, we can obtain an optimal QSP circuit as in Theorem \ref{thm:QSP}. 
\end{enumerate}

\section{Circuit depth optimization of general unitary synthesis}
\label{sec:unitary}

The following oracle is used in the circuit constructions in \cite{rosenthal2021query}.
\begin{definition}[Oracle $O_U$ of $U$]\label{def:oracle}
Let $U=[u_{y,x}]_{y,x\in\{0,1\}^n}\in\mathbb{C}^{2^n\times 2^n}$ denote a general $n$-qubit unitary operator. Let vector $u_x\in\mathbb{C}^{2^n}$ denote the $x$-th column of $U$ and  $\ket{u_x}=\sum_{y\in\{0,1\}^n} u_{y,x}\ket{y}$ is the corresponding $n$-qubit quantum state. The following unitary transformation $O_U$ is defined as the $U$-oracle:
\[\ket{x}\ket{0^n}\xrightarrow{O_U} \ket{x}\ket{u_x}, \text{for~all~}x\in\{0,1\}^n.\]
\end{definition}
This oracle {will be used} as an intermediate step of the circuit construction. Note that $O_U$ prepares the state $\ket{u_x}$ in the second register conditioned on the first register being $\ket{x}$. Given $O_U$, it is not immediate how to implement $U$, which changes $\ket{x}$ to $\ket{u_x}$ in place. However, we can indeed implement $U$ if we are allowed to use many queries to $O_U$ and $O_U^\dagger$.
{First, we can directly apply Theorem \ref{thm:multi-QSP} to obtain the following circuit construction for the oracle.} 
\begin{lemma}\label{lem:oracle}
For any $m\ge 0$ and $U\in\mathbb{C}^{2^n\times 2^n}$, the $U$-oracle $O_U$ and its inverse $O^\dagger_U$ can each be implemented by a standard quantum circuit of depth $O\left(n+\frac{4^n}{n+m}\right)$ and size $O(4^n)$ with $m$ ancillary qubits.
\end{lemma}
\paragraph{Remarks}
\begin{enumerate}
    \item {Ref.\cite{rosenthal2021query} gives a construction for $O_U$ and $O_U^\dagger$ with $O(n)$ depth using $m=\Theta(n4^n)$ ancillary qubits. In comparison, our Theorem \ref{thm:multi-QSP} only needs $m=\Theta(4^n/n)$ ancillary qubits to achieve $O(n)$ depth, and works for any $m \ge 0$.}
    \item {Since Theorem \ref{thm:multi-QSP} is tight for all values of parameters $(n,k,m)$, the bounds in Lemma \ref{lem:oracle} are also optimal.}
\end{enumerate}

\begin{lemma}[\cite{rosenthal2021query}]\label{lem:unitary_grover}
For any  $m\in [n,O(4^n/n)]$, any $n$-qubit unitary $U$ can be implemented by $\ell=O\left(2^{n/2}\right)$ many queries to oracle $O_U$ or $O_U^\dagger$:
\[\left(U\ket{\phi}\right)\ket{0}^{\otimes m} = C_\ell O_U^{(\dagger)}C_{\ell-1}O_U^{(\dagger)}C_{\ell-2}O_U^{(\dagger)}\cdots C_{2}O_U^{(\dagger)}C_{1}\left(\ket{\phi}\ket{0}^{\otimes m}\right),\text{ for all $n$-qubit states }\ket{\phi},\]
where $O_U^{(\dagger)}$ denotes either the oracle $O_U$ or $O_U^\dagger$, and each $C_i$ is a standard quantum circuit of depth $O(\log m)$ and size $O(m)$, and is independent of $U$.
\end{lemma}

By Lemma \ref{lem:oracle} and \ref{lem:unitary_grover}, we can obtain the following corollary.
\begin{corollary}\label{coro:general_US_sub}
For any $m\in [n,O(4^n/n)]$, any $n$-qubit unitary can be implemented by a standard quantum circuit of depth $O\left(n2^{n/2}+\frac{2^{5n/2}}{m}\right)$, using $m$ ancillary qubits.
\end{corollary}
\begin{proof}
    By Lemma \ref{lem:unitary_grover}, we obtain a quantum circuit of $\ell = O(2^{n/2})$ queries to $O_U$ and $O_U^\dagger$, with the total depth between queries is the summation of that of $C_i$'s, which is $\ell\cdot O(\log m)$. By Lemma \ref{lem:oracle}, each query to $O_U$ and $O_U^\dagger$ can be implemented by a circuit of depth $O(n+\frac{4^n}{(m-n)+n})$. Putting the two together, we obtain a circuit of depth
    \[O(2^{n/2})\cdot \big(\log m+n+\frac{4^n}{(m-n)+n}\big)= O\big(n2^{n/2}+\frac{2^{5n/2}}{m}\big).\]
\end{proof}
If the number of ancillary qubits is large, this bound improves the previous depth bound of $O(n2^{n})$ in \cite{sun2021asymptotically}. Next, we will show how to further improve the circuit depth for the parameter regime $\Omega(2^n) \le m \le O(4^n)$ by the cosine-sine decomposition and the following UCU. Note that each cosine-sine decomposition (Lemma \ref{lem:CSD}) reduces a general unitary to two $(1,n-1)$-UCUs and one $(n-1)$-UCG. One can continue this decomposition to further decrease the number of the target qubits to $n-2$, $n-3$, and so on. But it turns out that going all the way down to 1 target qubit does not give the most efficient construction. To see where to stop using the cosine-sine decomposition, we need to understand the circuit complexity of an  $(n-k,k)$-UCU for a general $k$, which is the subject of the next lemma.

\begin{lemma}\label{lem:general_UCU}
Let $T=\{n-k+1,n-k+2,\ldots, n\}$ and $S=[n-k]$ for any $k\in\{2,3,\ldots,n\}$. For any $m\ge 0$, any $(n-k,k)$-UCU $V^{S}_{T}$ can be implemented by a quantum circuit of depth $O\left(n2^{k/2}+\frac{2^{n+\frac{3}{2}k}}{n+m}\right)$ and size $O\left(m2^{k/2}+2^{n+3k/2}\right)$, with $m$ ancillary qubits.
\end{lemma}
\begin{proof}
We will implement $V^S_T$ by an $(n+m)$-qubit quantum circuit. When $m\le n$, we decompose $V^S_T$ into $O(2^k)$ $n$-qubit UCGs by repeatedly using cosine-sine decomposition in Lemma \ref{lem:CSD}. Combining with Lemma \ref{lem:UCG}, the total depth of $V^S_T$ is $O(2^k)\cdot O(n+2^n/(n+m))=O(2^{n+k}/(n+m))=$.
When $m\ge n$, the idea is to implement each $(n-k)$-qubit controlled $k$-qubit unitary $U_x$ by Lemma \ref{lem:unitary_grover}. Observe that these $(n-k)$ control qubits can be combined with the $k$ control qubits in the definition of oracle $O_{U_x}$ (Definition \ref{def:oracle}), to form an $(n,k)$-CQSP. Then we invoke Lemma \ref{lem:multi-QSP} to implement them and obtain the bounds.

According to Lemma \ref{lem:unitary_grover}, for all $x\in\{0,1\}^{n-k}$, any $k$-qubit unitary $U_x\in\mathbb{C}^{2^k\times 2^k}$ acting on qubits $\{n-k+1,n-k+2,\ldots,n\}$ can be implemented by $O(2^{k/2})$ queries to the $2k$-qubit oracles $O_{U_x}$ and $O^\dagger_{U_x}$. Here each $O_{U_x}$ is on $2k$ qubits, $k$ of which is for $U_x$ and the other $k$ using ancillary qubits. Using the notation $O_{U_x}^{(\dagger)}$ to denote oracle $O_{U_x}$ or $O^\dagger_{U_x}$, we have that
\begin{align*}
    (U_x\ket{\phi})\ket{0}^{\otimes m} & = C_{\ell}O_{U_x}^{(\dagger)} C_{\ell-1} O_{U_x}^{(\dagger)} \cdots  C_2 O_{U_x}^{(\dagger)} C_1 \left(\ket{\phi}\ket{0}^{\otimes m}\right), \text{for all $n$-qubit states} \ket{\phi}
\end{align*}
where $\ell=O(2^{k/2})$ and $C_1,\ldots,C_\ell$ are depth-$O(\log m)$ and size-$O(m)$ quantum circuits independent of $U_x$. 
Any $n$-qubit UCU $V^{S}_{T}$ can thus be implemented as follows
\begin{align}\label{eq:VST}
     V^{S}_{T}
    = &\sum_{x\in\{0,1\}^{n-k}}|x\rangle\langle x| \otimes U_x, \nonumber\\
    =&\sum_{x\in\{0,1\}^{n-k}}|x\rangle\langle x| \otimes \left[ C_{\ell}O_{U_x}^{(\dagger)} C_{\ell-1} O_{U_x}^{(\dagger)} \cdots  C_2 O_{U_x}^{(\dagger)} C_1\right], \nonumber\\
    = &[\mathbb{I}_{n-k}\otimes C_\ell]\big(\sum_{x\in\{0,1\}^{n-k}}|x\rangle\langle x| \otimes O_{U_x}^{(\dagger)} \big)\cdots [\mathbb{I}_{n-k}\otimes C_{2}] \big(\sum_{x\in\{0,1\}^{n-k}}|x\rangle\langle x| \otimes O_{U_x}^{(\dagger)} \big)[\mathbb{I}_{n-k}\otimes C_{1}],
\end{align}
where we switched the summation and multiplication again because of the block diagonal matrix as for Eq. \eqref{eq:switch}. Now we implement $\sum_{x\in\{0,1\}^{n-k}}|x\rangle\langle x| \otimes O_{U_x}^{(\dagger)}$. 
It can be regarded as a controlled quantum state preparation, which has $n$ controlled qubits and $k$ target qubits. Hence, by Lemma \ref{lem:multi-QSP}, we can implement it by a circuit of depth $O\left(n+k+\frac{2^{n+k}}{n+k+(m-k)}\right)=O\left(n+\frac{2^{n+k}}{m}\right)$ and size $O\left(2^{n+k}\right)$. 
Therefore by Eq. \eqref{eq:VST}, for any $m\in [n,O(4^n/n)]$, unitary $V^S_T$ can be realized by a circuit of depth \[O(2^{k/2})\cdot O\big(n+\frac{2^{n+k}}{m}\big)+O(2^{k/2})\cdot O(\log m)= O\big(n2^{k/2}+\frac{2^{n+\frac{3}{2}k}}{n+m}\big),\] and size \[O(2^{k/2})\cdot O\big(2^{n+k}\big)+O(2^{k/2})\cdot O(m) = O\big(m2^{k/2}+2^{n+3k/2}\big).\] 
\end{proof}


\paragraph{Remarks.}
\begin{enumerate}
    \item \textbf{Extension of UCG}. In \cite{sun2021asymptotically}, it was shown that any $n$-UCG can be implemented by a standard circuit of depth $O(n+2^n/(m+n))$. Lemma \ref{lem:general_UCU} generalizes this result to any $k$.
    \item \textbf{Tightness}. In \cite{shende2004minimal}, it was shown that any $n$-qubit unitary is determined by $4^n-1$ free parameters omitting a global phase. Because UCU $V^S_T$ in Lemma \ref{lem:general_UCU} are defined by $2^{n-k}$ different $k$-qubit unitaries, it is determined by $2^{n-k}\cdot 4^{k}-1=2^{n+k}-1$ free parameters. Similar to the depth lower bound for general unitary in \cite{sun2021asymptotically}, given $m$ ancillary qubits, the depth lower bound for UCU $V^S_T$ is $\Omega\left(\frac{2^{n+k}}{n+m}\right)$. Moreover, we can also obtain a depth lower bound $\Omega(n)$ by the light cone. This proof is the same as the proof of depth lower bound for quantum state preparation in \cite{sun2021asymptotically}. Combining the two results above, giving $m\ge 0$ ancillary qubits, the depth lower bound for UCU $V^S_T$ is $\Omega\left(n+\frac{2^{n+k}}{n+m}\right)$. When $k=O(1)$, the depth in Lemma \ref{lem:general_UCU} is asymptotically optimal. The case for general $k$ is left as an interesting open question.

\end{enumerate}

\begin{theorem}[Restatement of Theorem \ref{thm:unitary_mul_intro}]\label{thm:unitary_mul}
For any $m\ge 0$, any $n$-qubit unitary $U$ can be implemented by a standard quantum circuit with $m$ ancillary qubits of depth 
\[
\begin{cases}
           O\left(\frac{4^n}{n+m}\right), & \text{if~} m=O\left(\frac{2^n}{n}\right),  \\
          O\left(\frac{n^{1/2}2^{3n/2}}{m^{1/2}}\right),  & \text{if~} m\in[\Omega\left(2^n/n\right) , O(4^n/n)], \\
          O\left(n2^{n/2}\right), & \text{if~}  m = \Omega(\frac{4^n}{n}).
       \end{cases}
\]
\end{theorem}
\begin{proof}
Let $D_n(k,m)$ 
denote the {minimum} circuit depth 
of a general $n$-qubit UCU $V^{[n-k]}_{\{n-k+1,\ldots,n\}}$ with $m$ ancillary qubits. Especially, $D_n(n,m)$ denote the {minimum} depth of an $n$-qubit unitary $U$. 
According to the cosine-sine decomposition in Figure \ref{fig:cosine-sine}, for every $k\in[n]$ we have
\begin{align*}
    D_n(n,m)&\le 2D_n(n-1,m) + D_n(1,m) & (\text{Eq.}\eqref{eq:CSdecomp})\\
            &=2D_n(n-1,m)+O\big(n+\frac{2^n}{n+m}\big) & (\text{Lemma \ref{lem:UCG}})\\
            &\le 2^{n-k}D_n(k,m)+O\big(n2^{n-k}+\frac{2^{2n-k}}{n+m}\big).& (\text{by recursion})
\end{align*}
Now we use Lemma \ref{lem:general_UCU}, $D_n(k,m)=O\left(n2^{k/2}+\frac{2^{n+\frac{3}{2}k}}{n+m}\right)$. 
Hence, for any $k\in[n]$ we have
\begin{align*}
    &D_n(n,m)=2^{n-k}\times O\big(n2^{k/2}+\frac{2^{n+\frac{3}{2}k}}{n+m}\big)+O\big(n2^{n-k}+\frac{2^{2n-k}}{n+m}\big)=O\big(n2^{n-\frac{k}{2}}+\frac{2^{2n+\frac{1}{2}k}}{n+m}\big).
\end{align*}
When $m = O(2^n/n)$, we take $k = 1$, and get depth $D_n(n,m) \le O(4^n/(n+m))$. When $\Omega(2^n/n) \le m \le O(4^n/n)$, we take $k = \log m + \log n -n $, and get depth $D_n(n,m) \le O\left(\frac{n^{1/2}2^{3n/2}}{m^{1/2}}\right)$. When $m = \Omega(4^n/n)$, we take $k = n$, and get depth $D_n(n,m) \le O(n2^{n/2})$. This completes the proof. 
\end{proof}

\paragraph{Acknowledgments} We thank Jonathan Allcock for discussions on the background of QRAM. 

\bibliographystyle{quantum}
\bibliography{general_unitary}

\onecolumn

\end{document}